\documentclass[11pt]{article}

\usepackage[top=1in,bottom=1in,left=1in,right=1in]{geometry}
\usepackage[T1]{fontenc}
\usepackage{fullpage}
\usepackage{amsthm}

\usepackage{amssymb}
\usepackage{amsmath}
\usepackage{breqn}

\usepackage{enumitem}
\usepackage{subfigure}

\usepackage{algorithm}
\usepackage{algpseudocode}

\usepackage{xcolor}

\usepackage{float}
\usepackage{authblk}

\definecolor{DarkGreen}{rgb}{0,0.6,0}

\usepackage[colorlinks=true]{hyperref}

\newtheorem{theorem}{Theorem}
\newtheorem{lemma}{Lemma}[section]
\newtheorem{corollary}[lemma]{Corollary}
\newtheorem{invariant}[lemma]{Invariant}

\newtheorem{fact}[lemma]{Fact}

\theoremstyle{remark}
\newtheorem{remark}{Remark}

\algnewcommand{\AND}{\textbf{ and }}
\algnewcommand{\OR}{\textbf{ or }}

\allowdisplaybreaks

\newcommand{\poly}{\operatorname{poly}}
\newcommand{\polylog}{\operatorname{polylog}}
\newcommand{\rb}[2]{\raisebox{#1 mm}[0mm][0mm]{#2}}
\newcommand{\istrut}[2][0]{\rule[- #1 mm]{0mm}{#1 mm}\rule{0mm}{#2 mm}}

\algnewcommand\comment[1]{\hfill$\triangleright$ #1}

\newenvironment{enumerate2}
{\begin{enumerate}\setlength{\itemsep}{-0.05cm}\vspace{-0.1cm}}
{\vspace{-0.05cm}\end{enumerate}}

\def\twofigs #1{\hbox to \textwidth{#1}}

\newcommand{\EQ}{\mathsf{Eq}}
\newcommand{\SetD}{\mathsf{SetDisj}}
\newcommand{\SetI}{\mathsf{SetInt}}
\newcommand{\ExistsEQ}{\exists\mathsf{Eq}}
\newcommand{\perr}{p_{\mathrm{err}}}

\newcommand{\zero}[1]{\makebox[0mm][l]{$\displaystyle #1$}}

\newcommand{\SetDisjointness}{\textsf{SetDisjointness}}
\newcommand{\SetIntersection}{\textsf{SetIntersection}}
\newcommand{\EqualityTesting}{\textsf{EqualityTesting}}
\newcommand{\ExistsEqual}{\textsf{ExistsEqual}}
\newcommand{\CONGEST}{\textsf{CONGEST}}
\newcommand{\LOCAL}{\textsf{LOCAL}}

\newcommand{\dist}{\operatorname{dist}}

\newcommand{\M}{\mathcal{M}}
\newcommand{\hM}{\widehat{\mathcal{M}}}
\newcommand{\D}{\mathcal{D}}
\newcommand{\hD}{\widehat{\mathcal{D}}}
\newcommand{\tD}{\widetilde{\mathcal{D}}}
\newcommand{\X}{\mathcal{X}}
\newcommand{\Y}{\mathcal{Y}}
\newcommand{\W}{\mathcal{W}}
\newcommand{\Z}{\mathcal{Z}}
\newcommand{\N}{\mathcal{N}}
\newcommand{\E}{\mathrm{E}}
\newcommand{\supp}{\mathrm{supp}}
\newcommand{\ignore}[1]{}

\newcommand{\ID}{\mathsf{ID}}

\newcommand{\mi}[2]{\mathrm{I}(#1\;;\;#2)}
\newcommand{\h}[1]{\mathrm{H}(#1)}
\newcommand{\const}{16}
\newcommand{\res}{22}

\newcommand{\Szemeredi}{Szemer\'{e}di}

\newcommand{\Komlos}{Koml\'{o}s}
\newcommand{\Saglam}{Sa\u{g}lam}
\newcommand{\Hastad}{H\r{a}stad}

\begin{document}

\title{The Communication Complexity of Set Intersection\\ and Multiple Equality Testing\thanks{An extended abstract of this paper~\cite{HuangPZZ20} was presented at the 
31st ACM-SIAM Symposium on 
Discrete Algorithms (SODA 2020).  Compared to the extended abstract,
this manuscript contains a detailed and complete proof of the lower bound, as well as new upper bounds not appearing in~\cite{HuangPZZ20}
(Sections~\ref{sect:ub-et-2} and \ref{sect:ub-et-3}, pages 28--38). 
This work was supported by NSF grants CCF-1514383, CCF-1637546, and CCF-1815316.  Authors' emails: \texttt {\{hdawei, pettie\}@umich.edu}, \texttt {\{zhangyix16, zhijun-z16\}@mails.tsinghua.edu.cn}.}}

\author[1]{Dawei Huang}
\author[1]{Seth Pettie}
\author[2]{Yixiang Zhang}
\author[2]{Zhijun Zhang}

\affil[1]{University of Michigan}
\affil[2]{IIIS, Tsinghua University}

\date{}

\maketitle

\begin{abstract}
In this paper we explore fundamental problems in randomized communication complexity such as computing Set Intersection on sets of size $k$ 
and Equality Testing between vectors of length $k$.  
\Saglam{} and Tardos~\cite{SaglamT13} and Brody et al.~\cite{BrodyCKWY16} 
showed that for these types of problems, 
one can achieve optimal communication volume of $O(k)$ bits, 
with a randomized protocol that takes $O(\log^* k)$ rounds.  
They also proved~\cite{SaglamT13,BrodyCKWY16} 
that this is one point along the optimal round-communication tradeoff curve.

Aside from rounds and communication volume, 
there is a \emph{third} parameter of interest, namely the \emph{error probability} $\perr$, which we write $2^{-E}$. 
It is straightforward to show that protocols for Set Intersection
or Equality Testing need to send 
at least $\Omega(k + E)$ bits, regardless of the number of rounds.  
Is it possible to simultaneously achieve optimality in 
all three parameters, namely 
$O(k + E)$ communication and $O(\log^* k)$ rounds?

In this paper we prove that there is no universally optimal algorithm,
and complement the existing round-communication tradeoffs~\cite{SaglamT13,BrodyCKWY16}
with a new tradeoff between rounds, communication, 
and probability of error.  In particular:
\begin{itemize}
    \item Any protocol for solving Multiple Equality Testing in $r$ rounds with failure probability $\perr = 2^{-E}$ 
    has communication volume $\Omega(Ek^{1/r})$.  
    
    \item We present several algorithms for Multiple Equality Testing (and its variants) that match or nearly match our lower bound and 
    the lower bound of ~\cite{SaglamT13,BrodyCKWY16}.

    \item Lower bounds on Equality Testing extend to Set Intersection, for every $r, k,$ and $\perr$ (which is trivial); in the reverse direction, we prove upper bounds on Equality Testing for $r, k, \perr$ imply similar upper bounds on Set Intersection with parameters $r+1, k,$ and $\perr$.
\end{itemize}

Our original motivation for considering $\perr$ as an independent
parameter came from the problem of enumerating triangles in distributed ($\CONGEST$) networks having maximum degree $\Delta$.
We prove that this problem can be solved in 
$O(\Delta/\log n + \log\log \Delta)$ time with high 
probability $1-1/\poly(n)$.
This beats the trivial (deterministic)
$O(\Delta)$-time algorithm and is superior to the $\tilde{O}(n^{1/3})$ algorithm of~\cite{ChangPZ19,ChangS19} when $\Delta=\tilde{O}(n^{1/3})$. 
\end{abstract}

\newpage

\section{Introduction} \label{sect:intro}

Communication Complexity was defined by Yao~\cite{Yao79} in 1979 and 
has become an indispensible tool for proving lower bounds in
models of computation in which the notions of 
\emph{parties} and \emph{communication} are not direct.
See, e.g., books and monographs~\cite{Roughgarden16,RaoY19,KushilevitzN97} 
and surveys~\cite{ChattopadhyayP10,Lovasz89} on the subject.
In this paper we consider some of the most fundamental and well-studied 
problems in this model, such as $\SetDisjointness$, $\SetIntersection$,
$\ExistsEqual$, and $\EqualityTesting$.  Let us briefly define these problems
formally since the terminology is not completely standard.

\begin{description}
\item[$\SetDisjointness$ and $\SetIntersection$.]
In the $\SetDisjointness$ problem 
Alice and Bob receive sets $A\subset U$ and $B\subset U$ where
$|A|, |B| {} \leq k$ and must determine whether 
$A\cap B = \emptyset$.
Define $\SetD(k,r,\perr)$ to be the minimum 
communication complexity of an $r$-round randomized protocol 
for this problem that errs with 
probability at most $\perr$.  
We can assume that
$|U|{}=O(k^2/\perr)$ without loss of generality.\footnote{Before the first round of communication, pick a pairwise independent 
$h : U\mapsto [O(k^2 / \perr)]$ 
and check whether $h(A)\cap h(B) = \emptyset$ with error probability $\perr/2$.
Thus, having $\SetD$ depend additionally on $|U|$ is somewhat redundant, at least when $|U|$ is large.}
The input to the $\SetIntersection$ problem is the same, 
except that the parties 
must \emph{report the entire set $A\cap B$}.
Define $\SetI(k,r,\perr)$ to be the minimum 
communication complexity of an $r$-round protocol 
for $\SetIntersection$.

\item[$\EqualityTesting$ and $\ExistsEqual$.] 
In the $\EqualityTesting$ problem 
Alice and Bob hold vectors $\mathbf{x}\in U^k$ and $\mathbf{y}\in U^k$ 
and must determine, for each index $i\in [k]$, whether $x_i=y_i$ or $x_i\neq y_i$.
A potentially easier version of the problem,
$\ExistsEqual$, is to determine if there \emph{exists} at least one 
index $i\in [k]$ for which $x_i=y_i$.  
Define $\EQ(k,r,\perr)$ to be the randomized communication complexity 
of any $r$-round protocol for 
$\EqualityTesting$ that errs 
with probability $\perr$, 
and $\ExistsEQ(k,r,\perr)$ the corresponding
complexity of $\ExistsEqual$. 
Once again, we can assume that $|U|{}=O(k/\perr)$ 
without loss of generality.
\end{description}

The deterministic communication complexity of these problems is well understood~\cite{KushilevitzN97}.  
(The optimal protocol is for Alice to send her entire input to Bob.)
Although the randomized communication complexity of these problems has been studied 
extensively~\cite{HastadW07,Razborov92,KalyanasundaramS92,FederKNN95,DasguptaKS12,Nikishkin13,BuhrmanGMW13,SaglamT13,BrodyCKWY16}, 
most prior work has focused on the relationship 
between \emph{round complexity}
and \emph{communication volume}, 
and has usually not treated $\perr$ as a parameter independent of $k$.

\paragraph{History.}
\Hastad{} and Wigderson~\cite{HastadW07} gave 
an $O(\log k)$-round 
protocol for $\SetDisjointness$ 
in which Alice and Bob communicate $O(k)$ bits, which
matched an $\Omega(k)$ lower bound
of Kalyanasundaram and Schnitger~\cite{KalyanasundaramS92}; see also~\cite{Razborov92,BuhrmanGMW13,DasguptaKS12}.
Feder et al.~\cite{FederKNN95} proved that $\EqualityTesting$ can be solved with $O(k)$ communication 
by an $O(\sqrt{k})$-round protocol that errs with probability
$\exp(-\sqrt{k})$.
Nikishkin~\cite{Nikishkin13} later improved their round complexity 
and error probability to $\log k$ and $\exp(-k/\polylog(k))$, respectively. 
Improving~\cite{HastadW07}, \Saglam{} and Tardos~\cite{SaglamT13}
gave an $r$-round protocol for $\SetDisjointness$
that uses $O(k\log^{(r)} k)$ communication, where
$\log^{(r)}$ is the $r$-fold iterated logarithm function.  
For $r = \log^* k$ the error probability of this protocol is $\exp(-\sqrt{k})$, coincidentally matching~\cite{FederKNN95}.  
In independent work, Brody et al.~\cite{BrodyCKWY16} gave $r$-round and $O(r)$-round
protocols for $\ExistsEqual$ and $\SetIntersection$, respectively, 
that use $O(k\log^{(r)} k)$ communication and err with probability $1/\poly(k)$.

\Saglam{} and Tardos~\cite{SaglamT13} were the first to show that this $O(k\log^{(r)}k)$ round vs.~communication tradeoff is optimal, using a combinatorial round elimination technique. 
In particular, this lower bound applies to any $\ExistsEqual$ protocol even with constant error probability.
Independently, Brody et al.~\cite{BrodyCKWY14,BrodyCKWY16}
established the same lower bound tradeoff for $\ExistsEqual$,
assuming the error probability is at most $1/\poly(k)$.
Brody et al.~\cite{BrodyCKWY16} also introduced a \emph{randomized} reduction 
from $\SetIntersection$ to $\EqualityTesting$, which errs with probability 
$\exp(-\tilde{O}(\sqrt{k}))$, i.e., it cannot be applied when the desired total error
probability $\perr$ is sufficiently small.

\begin{table}
    \centering
    \begin{tabular}{l|l|l|l|l}
{\bf Problem}             & {\bf Commun.} & {\bf Rounds}        & {\bf Error Probability} & {\bf Notes}\\\hline\hline
\EqualityTesting    & $O(k)$        & $O(\sqrt{k})$ & $\exp(-\sqrt{k})$ & \cite{FederKNN95}\\\hline
\EqualityTesting    & $O(k)$        & $\log k$ & $\exp(-k/\polylog(k))$ & \cite{Nikishkin13}\\\hline
\SetDisjointness    & $O(k)$        & $O(\log k)$   &   Constant    & \cite{HastadW07}\\\hline
\SetDisjointness    & $O(k\log^{(r)} k)$        & $r$   & $\ge \exp(-\sqrt{k})$ & \cite{SaglamT13}\\\hline
\ExistsEqual        &                           & $r$   &       & \\\cline{1-1}\cline{3-3}
\SetIntersection    & \rb{2.5}{$O(k\log^{(r)} k)$}        & $O(r)$ & \rb{2.5}{$1/\poly(k)$}  & \rb{2.5}{\cite{BrodyCKWY16}}\\\hline
\ExistsEqual        &&&&\\
\SetDisjointness   & \rb{2.5}{$O(k + Ek^{1/r})$}     & \rb{2.5}{$r + \log^*(k/E)$}    &    & \\\cline{1-3}
\EqualityTesting   & $O(k + Ek^{1/r}\cdot \log r$  &               & \ \ \ $2^{-E}$   & {\bf new} \\  
\ \ \ \ \ and        & \hfill $ +\, Er\log r)$   & \rb{2.5}{$r + \log^*(k/E) \;\; [+1]$}  &                     & \\\cline{2-3}
{[\SetIntersection]}  & $O(k+E)$  &  $\log k$ &   & \\\cline{1-5}
\multicolumn{5}{l}{\bf Lower Bounds\istrut{7}}\\\hline
\SetDisjointness    & $\Omega(\sqrt{k})$    & $\infty$ &   Constant    & \cite{BabaiFS86}\\\hline
\SetDisjointness    & $\Omega(k)$           & $\infty$ &   Constant    & \cite{KalyanasundaramS92}\\\hline
\ExistsEqual    & $\Omega(k\log^{(r)} k)$ & $r$     &  Constant     & \cite{SaglamT13}\\\hline
\ExistsEqual    & $\Omega(k\log^{(r)} k)$ & $r$     &  $1/\poly(k)$ & \cite{BrodyCKWY16}\\\hline
\ExistsEqual    & $\Omega(Ek^{1/r})$  &   $r$     &   $2^{-E}$ & {\bf new}\\\hline\hline
    \end{tabular}
    \caption{Upper and Lower bounds on $\SetDisjointness, \SetIntersection, \EqualityTesting,$ and $\ExistsEqual$.  Via trivial reductions, lower bounds on $\ExistsEqual$ extend to all four problems, and upper bounds on $\SetIntersection$ extend to all four problems.  
    From Theorem~\ref{thm:Equivalence-Int-Eq},
    the upper bounds on $\SetIntersection$ and $\SetDisjointness$ follow from those of 
    $\EqualityTesting$
    and $\ExistsEqual$, respectively, 
    $+1$ round of communication.
    The log-star function is defined as 
    $\log^*(x) = \min\{i : \log^{(i)}(x) \leq 1\}$, e.g.,
    $\log^*(k/E)=0$ if $E\ge k$.
    }
    \label{tab:history}
\end{table}

\subsection{Contributions}

First, we observe that a simple \emph{deterministic} reduction
shows that, up to one round of communication, 
$\SetIntersection$ is equivalent to $\EqualityTesting$ for any $\perr$,
and $\SetDisjointness$ is equivalent to $\ExistsEqual$ for any $\perr$.  
Theorem~\ref{thm:Equivalence-Int-Eq} is proved in Appendix~\ref{sect:reductions}; 
it is inspired by the randomized reduction of Brody et al.~\cite{BrodyCKWY16}.

\begin{theorem}\label{thm:Equivalence-Int-Eq}
For any parameters $k\ge 1,r\ge 1,$ and $\perr = 2^{-E}$, 
it holds that
\begin{align*}
\EQ(k,r,\perr) &\leq \SetI(k,r,\perr), 
& \SetI(k,r+1,\perr) &\leq \EQ(k,r,\perr) + \zeta,\\
\ExistsEQ(k,r,\perr) &\leq \SetD(k,r,\perr), 
&\SetD(k,r+1,\perr) &\leq \ExistsEQ(k,r,\perr) +\zeta,
\end{align*}
where $\zeta = O(k + \log E)$.
\end{theorem}

Second, we prove that in any of the four problems, 
it is impossible to simultaneously achieve communication volume 
$O(k + E)$ in $O(\log^* k)$ rounds for all $k, \perr=2^{-E}$.
Specifically, any $r$-round protocol needs $\Omega(Ek^{1/r})$ communication.
Whereas the implication of \cite{SaglamT13,BrodyCKWY16} is that optimal $O(k)$
communication is only possible with $\Omega(\log^* k)$ rounds, the implication of our
work is that optimal communication $O(k+E)$ is 
only possible with $\Omega(\log k)$ rounds, whenever $E\ge k$.

We complement our lower bounds with matching or nearly matching upper bounds.
First, we show that in any $\EqualityTesting$/$\ExistsEqual$ instance with 
$E < k$, one can, with probability $1-2^{-\Theta(E)}$, 
reduce the effective number of coordinates to $E$ using $O(k)$ communication and $\log^*(k/E)$ rounds.
Thus, we can simplify the following discussion by assuming that $E\ge k$.

We give \emph{four} distinct protocols, the first of which
solves $\EqualityTesting$ with $O(rEk^{1/r})$ communication, 
which is optimal whenever $r=O(1)$.
The remaining three protocols attempt to get rid of this extraneous
$r$ factor in different situations.  Our second protocol shows
that it \emph{is} possible to achieve $O(Ek^{1/r})$ complexity, 
but for the slightly simpler problem of $\ExistsEqual$. 
Our third protocol shows that with $O(r)$ rounds (instead of $r$ rounds)
it is possible to achieve $O(Ek^{1/r})$ communication.  In particular,
absolutely optimal communication $O(E)$ is possible with $\log k = O(r)$ communication.

Our first three protocols show that the optimal round-communication-error
tradeoff for $\EqualityTesting$ can be achieved whenever $r=O(1)$ or $r=\Omega(\log k)$,
or for any $r$ in the case of $\ExistsEqual$.
The remaining problem ($\EqualityTesting$ in $r$ rounds, $r$ between $\omega(1)$ and $o(\log k)$) 
seems to be quite difficult.  Our fourth protocol solves $\EqualityTesting$ with
$O(Ek^{1/r}\log r + Er\log r)$ communication, which for $r \in [1,\log k/\log\log k]$ is dominated
by the first term and therefore within a $\log r \leq \log\log k$ factor of optimal.
A close study of our second and fourth protocols reveals a key distinction between 
$\EqualityTesting$ from $\ExistsEqual$, which is only relevant
when the probability of error is quite small (e.g., $E\ge k$).  It is \emph{plausible} that 
$\EqualityTesting$ is asymptotically harder than $\ExistsEqual$ for many values of $r$,
and for similar reasons, that $\SetIntersection$ is asymptotically harder than $\SetDisjointness$.

\medskip

Our original interest in $\SetIntersection$ came from distributed subgraph detection in $\CONGEST$\footnote{In the $\CONGEST$ model there is a graph $G=(V,E)$ whose vertices are identified with processors and whose edges represent bidirectional communication links.  Each vertex $v$ does not know $G$, and is only initially aware of an $O(\log n)$-bit $\ID(v)$, 
$\deg(v)$, and global parameters $n \geq {}|V|$ and $\Delta \geq \max_{u\in V}\deg(u)$.  Communication proceeds in synchronized rounds; in each round, each processor can send a (different) $O(\log n)$-bit message to each of its neighbors.} networks, which has garnered significant interest in recent years~\cite{ChangS19,ChangPZ19,IzumiG17,AbboudCKL17,DruckerKO14,KorhonenR2018,FischerGKO18,CzumajK18,GonenO18}.
Izumi and LeGall~\cite{IzumiG17} proved that triangle enumeration\footnote{Every triangle (3-cycle) in $G$ must be reported by \emph{some} vertex.}
requires $\Omega(n^{1/3}/\log n)$ rounds in the $\CONGEST$ model, 
and further showed that \emph{local} triangle enumeration\footnote{Every triangle in $G$ must be reported by at least one of the three constituent vertices.  Izumi and LeGall~\cite{IzumiG17} only stated the $\Omega(n/\log n)$ lower bound but it can also 
be expressed in terms of $\Delta$.} requires $\Omega(\Delta/\log n)$ rounds in $\CONGEST$, which can be as large as $\Omega(n/\log n)$.

The most natural way to solve (local) triangle enumeration is, 
for every edge $\{u,v\} \in E(G)$, 
to have $u$ and $v$ run a 
two-party $\SetIntersection$ protocol
in which they compute $N(u)\cap N(v)$,
where $N(u) = \{\ID(x) \mid \{u,x\}\in E(G)\}$ and 
$\ID(x) \in \{0,1\}^{O(\log n)}$ is $x$'s unique identifier.
Any $r$-round protocol with communication volume
$O(\Delta)$ can be simulated in $\CONGEST$
in $O(\Delta/\log n + r)$ rounds since the message size is $O(\log n)$ bits.  
However, to guarantee a \emph{global} probability of 
success at least $1-1/\poly(n)$, the failure probability
of each $\SetIntersection$ instance must be $\perr = 2^{-E}$, $E=\Theta(\log n)$, 
which is independent of $\Delta$.
Our communication complexity 
lower bound suggests that to achieve this error probability, 
we would need $\Omega((\Delta + E\Delta^{1/r})/\log n + r)$ 
$\CONGEST$ rounds,
i.e., with $r=\log\Delta$ 
we should not be able to do better
than 
$O(\Delta/\log n + \log\Delta)$.
We prove that (local) triangle enumeration can 
actually be solved \emph{exponentially} faster,
in $O(\Delta/\log n + \log\log\Delta)$ $\CONGEST$ rounds,
without \emph{necessarily} solving every 
$\SetIntersection$ instance.

\paragraph{Organization.}
The proof of Theorem~\ref{thm:Equivalence-Int-Eq}
on the near-equivalence of $\SetIntersection/\SetDisjointness$ 
and $\EqualityTesting/\ExistsEqual$ appears in Appendix~\ref{sect:reductions}.
Section~\ref{sect:prelim} reviews concepts from information theory and communication complexity. 
In Section~\ref{sect:lb} we present new lower bounds
for both $\EqualityTesting$ and $\ExistsEqual$ that incorporate 
rounds, communication, and error probability.
Section~\ref{sect:ub} presents nearly matching upper bounds for $\EqualityTesting$ and $\ExistsEqual$,
and Section~\ref{sect:triangle} applies them to the 
distributed triangle enumeration problem.
We conclude with some open problems in Section~\ref{sect:conclusion}.

\section{Preliminaries} \label{sect:prelim}

\subsection{Notational Conventions}
The set of positive integers at most $t$ is denoted $[t]$.
Random variables are typically written as capital letters ($X,Y,M$, etc.) and the values they take on are lower case ($x,y,m$, etc.).
The letters $p, q, \mu, \D$ are reserved for 
probability mass functions (\emph{p.m.f.}). 
E.g., $\D(x)$ denotes the probability that $X = x$ whenever 
$X\sim \D$.
The \emph{support} $\supp(\D)$ of a distribution $\D$
is the  set of all $x$ for which $\D(x)>0$.
If $\X \subseteq \supp(\D)$, $\D(\X) = \sum_{x \in \X}\D(x)$. 

Many of our random variables are vectors.
If $x$ is a $k$-dimensional vector and $I\subseteq [k]$,
$x_I$ is the projection of $x$ onto the coordinates in $I$
and $x_i$ is short for $x_{\{i\}}$. 
Similarly, if $\D$ is the p.m.f. of a $k$-dimensional 
random variable, $\D_I$ is the marginal distribution
of $\D$ on the index set $I\subseteq [k]$.

Throughout the paper, 
$\log$ and $\exp$ are the 
base-2 logarithm and exponential functions,
and $\log^{(r)}$ and $\exp^{(r)}$ their 
$r$-fold iterated versions:
\[
\log^{(0)}(x) = \exp^{(0)}(x) = x, \quad
\log^{(r)}(x) = \log(\log^{(r-1)}(x)), \quad
\exp^{(r)}(x) = \exp(\exp^{(r-1)}(x)).
\]
The log-star function is defined to be
$\log^*(x) = \min\{r \mid \log^{(r)}(x) \le 1\}$.
In particular, $\log^*(x)=0$ if $x\leq 1$.

\subsection{Information Theory}

The most fundamental concept in information theory is \emph{Shannon entropy}. The Shannon entropy of a discrete random variable $X$ is defined as
\[
\h{X} = -\sum_{x \in \supp(X)} \Pr[X=x] \log \Pr[X=x].
\]
Since there may be cases in which different distributions are defined for the ``same'' random variable, 
we use $\h{p}$ in place of $\h{X}$ if $X$ is drawn from a p.m.f. $p$.  We also write 
$\h{\alpha}$, $\alpha\in (0,1)$, to be the entropy of a Bernoulli random variable with success probability $\alpha$.
In general, we freely use a random variable and its p.m.f.~interchangeably.

The \emph{joint entropy} $\h{X,Y}$ of two random variables $X$ and $Y$ is simply
\[
\h{X,Y} = -\sum_{x \in \supp(X)} \sum_{y \in \supp(Y)} \Pr[X=x \land Y=y] \log \Pr[X=x \land Y=y].
\]
This notion can be easily extended to cases of more than two random variables. Here, we state a well known fact about joint entropy.

\begin{fact}\label{fac:joint-entropy}
For any random variables $X_1,X_2,\ldots,X_n$, their joint entropy is at most the sum of their individual entropies, i.e.,
$\h{X_1,X_2,\ldots,X_n} \le \sum_{i=1}^n \h{X_i}$.
\end{fact}

The \emph{conditional entropy} of $Y$ conditioned on another random variable $X$, denoted $\h{Y \mid X}$, measures the expected amount of extra information required to fully describe $Y$ if $X$ is known. It is defined to be
\begin{align*}
\h{Y \mid X}
& = \h{X,Y} - \h{X} \\
& = -\sum_{x \in \supp(X)} \Pr[X=x] \sum_{y \in \supp(Y)} \Pr[Y=y \mid X=x] \log \Pr[Y=y \mid X=x] \:\ge\: 0,
\end{align*}
which can be viewed as a weighted sum of entropies of a number of conditional distributions.

Finally, the \emph{mutual information} $\mi{X}{Y}$ between two random variables $X$ and $Y$ quantifies the amount of information that is revealed about one random variable through knowing the other one:
\begin{align*}
\mi{X}{Y}
& = \h{X} - \h{X \mid Y} \\
& = \h{X} + \sum_{y \in \supp(Y)} \Pr[Y=y] \sum_{x \in \supp(X)} \Pr[X=x \mid Y=y] \log \Pr[X=x \mid Y=y].
\end{align*}

\subsection{Communication Complexity}

Let $f(x, y)$ be a function over domain $\X \times \Y$, and consider any two-party communication protocol $Q(x, y)$ 
that computes $f(x, y)$, where one party holds $x$ and the other holds $y$. 
The \emph{transcript} of $Q$ on $(x, y)$ is defined to be the concatenation of all messages exchanged by the two parties, 
in order, as they execute on input $(x, y)$. 
The \emph{communication cost} of $Q$ is the maximum transcript length produced by $Q$ over all possible inputs.

Let $Q_d$ be a deterministic protocol for $f$ and suppose $\mu$ 
is a distribution over $\X \times \Y$. 
The \emph{distributional error probability of $Q_d$ with respect to $\mu$} is the probability $\Pr_{(x, y) \sim \mu}[Q_d(x, y) \neq f(x, y)]$. For any $0 < \epsilon < 1$, the \emph{$(\mu, \epsilon)$-distributional deterministic communication complexity} of the function $f$ is the minimum communication cost of any protocol $Q_d$ that has distributional error probability at 
most $\epsilon$ with respect to the distribution $\mu$. 

A randomized protocol $Q_r(x, y, w)$ also takes a public random string $w \sim \W$ as input. The error probability of $Q_r$ is calculated as $\max_{(x, y) \in \X \times \Y}\Pr_{w \sim \W}[Q_r(x, y, w) \neq f(x, y)]$. The \emph{$\epsilon$-randomized communication complexity} of $f$ is the minimum communication cost of $Q_r$ over all protocols $Q_r$ with error probability at most $\epsilon$.

Yao's \emph{minimax principle}~\cite{Yao77} is a common starting point for lower bound proofs in randomized communication complexity. The easy direction of Yao's minimax principle states that the communication cost of the best deterministic protocol specific to any particular distribution is at most the communication cost of any randomized protocol on its worst case input. 

\begin{lemma}[Yao's minimax principle~\cite{Yao77}]
Let $f:\X \times \Y \mapsto \Z$ be the function to be computed.
Let $D_{\mu, \epsilon}(f)$ be the $(\mu, \epsilon)$-distributional deterministic communication complexity of $f$, and let $R_{\epsilon}(f)$ be the $\epsilon$-randomized communication complexity of $f$. Then for any $0 < \epsilon < 1/2$, 
\[
\max_\mu D_{\mu, \epsilon}(f) \leq R_{\epsilon}(f).
\]
\end{lemma}

Therefore, to show a lower bound on the $\epsilon$-randomized communication complexity of a function $f$, it suffices to find a hard distribution $\mu$ on the input set and prove a lower bound for the communication cost of any deterministic protocol that has distributional error probability at most $\epsilon$ with respect to $\mu$.

\section{Lower Bounds on $\ExistsEqual$ and $\EqualityTesting$} \label{sect:lb}

In this section we prove lower bounds on $\EqualityTesting$ and $\ExistsEqual$.  Theorem~\ref{thm:lb-et} obviously follows directly from Theorem~\ref{thm:lb-ee}, 
but we prove them in that order nonetheless because Theorem~\ref{thm:lb-et} 
is a bit simpler.

\begin{theorem}\label{thm:lb-et}
Any $r$-round randomized protocol for $\EqualityTesting$ on vectors of length $k$
that errs with probability $\perr = 2^{-E}$ 
requires at least $\Omega(Ek^{1/r})$ bits of communication.
\end{theorem}

\begin{theorem}\label{thm:lb-ee}
Any $r$-round randomized protocol for $\ExistsEqual$ on vectors of length $k$
that errs with probability $\perr = 2^{-E}$
requires at least $\Omega(Ek^{1/r})$ bits of communication.
\end{theorem}

Without any constraint on the number of rounds, 
$\EqualityTesting$ trivially requires $\Omega(k)$ communication.
$\ExistsEqual$ also requires $\Omega(k)$ communication,
through a small modification to the $\SetDisjointness$ lower bounds~\cite{KalyanasundaramS92,Razborov92}.
Even when $k=1$, we need at least
$\Omega(E)$ communication to solve 
$\EqualityTesting/\ExistsEqual$ with error probability $2^{-E}$~\cite{KushilevitzN97}.
Thus, we can assume that $E = \Omega(k^{1-1/r})$, 
$k^{1/r} = \Omega(1)$, and hence $r=O(\log k)$. 
For example, some calculations later 
in our proof hold when $r \le (\log k)/6$.
When proving Theorem~\ref{thm:lb-ee}, we will further assume $E = \Omega(\log k)$ when $r=1$,
which is reasonable because of {\Saglam} and Tardos' 
$\Omega(k\log^{(r)}k) = \Omega(k\log k)$ lower bound~\cite{SaglamT13}.

\subsection{Structure of the Proof}\label{sect:desc-lb}

We consider deterministic strategies for $\ExistsEqual/\EqualityTesting$
when Alice and Bob 
pick their input vectors independently 
from the uniform distribution on $[t]^k$, 
where $t = 2^{cE}$ and $c=1/2$.
Although the probability of seeing a collision 
in any particular coordinate is small, it is still
much larger than the tolerable error probability (since $c<1$),
so it is incorrect to declare ``not equal in every coordinate''
without performing any communication.

We suppose, for the purpose of obtaining a contradiction, that there is a protocol for $\EqualityTesting$ 
with error probability $2^{-E}$ and communication complexity 
$c'Ek^{1/r}$, where $c' = c/100$.  
The length of the $j$th message 
is $l_j$, which could depend 
on the parameters ($E,r,k$, etc.) and possibly in some complicated way on the transcript of the protocol before round $j$.\footnote{In the context of $\ExistsEqual/\EqualityTesting$, it is natural to think about uniform-length messages, $l_j = c'Ek^{1/r}/r$,
or lengths that decay according to some convergent series, 
e.g., $l_j \propto c'Ek^{1/r}/2^j$ or $l_j \propto c'Ek^{1/r}/j^2$.}

Our proof must necessarily consider transcripts of the protocol that are extremely unlikely (occurring with probability close to $2^{-E}$) and also maintain a high level of uncertainty about \emph{which} coordinates of Alice's and Bob's vectors might be equal.  Consider the first message.  
Alice picks her input vector $x\in [t]^k$, 
which dictates the first message $m_1$.  Suppose, for simplicity,
that it betrays exactly $l_1/k < c'Ek^{1/r - 1}$ bits of information 
per coordinate of $x$.  Before Bob can respond with a message $m_2$ he must commit to his input, say $y$.  Most values of $y$ result in ``good'' outcomes: nearly all non-equal coordinates get detected immediately and the effective size of the problem is dramatically reduced.  We are not interested in these values of $y$, only very ``bad'' values. 
Let $I_1$ be the first $k^{1-1/r}$ coordinates (or, more generally, $k^{1-1/r}$ coordinates that $m_1$ revealed below-average information about).  With probability about
$(2^{-c'Ek^{1/r-1}})^{|I_1|} = 2^{-c'E}$, Bob picks an input
$y$ that is \emph{completely consistent} with Alice's on $I_1$,
i.e., as far as he can tell $y_i = x_i$ for every $i\in I_1$.
Rather than sample $y$ uniformly from $[t]^k$, we sample 
it from a ``hybrid'' distribution: $y_{I_1}$ is sampled from
the same distribution that $m_1$ revealed about $x_{I_1}$ 
(forcing the above event to happen with probability 1), and 
$y_{[k]\backslash I_1}$ is sampled from Bob's former 
distribution (in this case, the uniform distribution on $[t]^{k-{}|I_1|}$), conditioned on the value of $y_{I_1}$.

This process continues round by round.  Bob's message $m_2$
betrays at most 
$l_2/|I_1| {} < c'Ek^{2/r-1}$ bits of information on each coordinate of $y_{I_1}$, and there must be an index 
set $I_2 \subset I_1$ with $|I_2|{}=k^{1-2/r}$ such that, 
with probability around $2^{-c'E}$,
it is completely consistent that
$x_{I_2} = y_{I_2}$.
Alice resamples her input so that this
(rare) event occurs with probability 1, 
generates $m_3$, and continues.

At the end of this process $|I_r|{}=k^{1-r/r}=1$, and yet Alice and Bob have revealed less than the full $cE$ bits of entropy about $x_{I_r}$ and $y_{I_r}$.  Regardless of whether they report ``equal'' or ``not equal'' (on $I_r$), they are wrong with probability greater than $2^{-E}$.  Are we done?  Absolutely not!
The problem is that this strange process for sampling a possible transcript of the protocol might itself only find transcripts that occur with probability $\ll 2^{-E}$, making any conclusions we make about its (probability of) correctness moot.
Generally speaking, 
we need to show that Alice's and Bob's actions are consistent 
with events that occur with probability $\gg 2^{-E}$.

Let us first make every step of the above process a bit more formal.  It is helpful to think about Alice's and Bob's inputs not being \emph{fixed} vectors selected at time zero, 
but simply distributions over vectors that change as messages progressively reveal more information about them.
\begin{itemize}
    \item Before the $j$th round of communication, 
    the sender of the $j$th message's input is drawn 
    from a discrete distribution
    $\hD^{(j-1)}$ over $[t]^k$.  The receiver of the $j$th message's input is drawn from the distribution $\D^{(j-1)}$. 
    For example, when $j=1$, if Alice speaks first then
    her initial distribution, $\hD^{(0)}$, and Bob's initial distribution, $\D^{(0)}$,
    are both uniform over $[t]^k$.
    \item Before the $j$th round of communication both parties are aware of an index set $I_{j-1}$ such that, informally,
    (i) the distributions $\D^{(j-1)}_{I_{j-1}}$
    and $\hD^{(j-1)}_{I_{j-1}}$ are very similar, and in particular, it is consistent that their inputs are identical on $I_{j-1}$, and
    (ii) the messages transmitted so far reveal ``average'' or below-average information about these coordinates.
    For example, $I_0 = [k]$ and it is consistent with the empty transcript that Alice's and Bob's inputs are identical on every coordinate.  
    
    \item The $j$th message is a random variable $M_j \in \{0,1\}^{l_j}$. 
    In order to pick an $m_j$ according to the right distribution, the sender picks an input $x \sim \hD^{(j-1)}$
    which, together with the history $m_1,\ldots,m_{j-1}$,
    determines $m_j$.  The sender transmits $m_j$ to the receiver and promptly forgets $x$.  
    The sender's new distribution (i.e., $\hD^{(j-1)}$, conditioned on $M_j = m_j$) 
    is called $\D^{(j)}$.
    
    \item The distribution $\D^{(j)}$ may reveal information about the coordinates $I_{j-1}$ in an irregular fashion.  We find a subset $I_j \subset I_{j-1}$ of coordinates, $|I_j|{}=k^{1-j/r}$, for which the amount of information revealed by $\D^{(j)}_{I_j}$ is at most average.  The receiver of $m_j$ changes his input distribution to $\hD^{(j)}$, which is defined so that it
    basically agrees with $\D^{(j)}_{I_j}$ and
    the marginal distribution $\hD^{(j)}_{[k]\backslash I_j}$,
    conditioned on the value selected by $\D^{(j)}_{I_j}$,
    is identical to $\D^{(j-1)}_{[k]\backslash I_{j}}$.
    
    \item The reason $\D^{(j)}_{I_j}$ and $\hD^{(j)}_{I_j}$ are not \emph{identical} is due to two filtering steps.  To generate $\hD^{(j)}$, we remove points from the support that have tiny (but non-zero) probability, which may be too close 
    to the error probability.  Intuitively these rare events necessarily represent a small fraction of the probability mass.  Second, we remove points from the support if the ratio of their probability occurring under $\D^{(j)}$ over $\D^{(j-1)}$ is too high.  Intuitively, we want to conclude
    that if there is a high probability of an error occurring under $\D^{(j)}$ then the probability is also high under $\D^{(j-1)}$ (and by unrolling this further, under $\D^{(0)}$).
    This argument only works if the ratios are what we would expect, given how much information is being revealed about these coordinates by $m_j$.
    As a result of these two filtering steps, $\D^{(j)}_{I_j}(x_{I_j})$ and $\hD^{(j)}_{I_j}(x_{I_j})$ differ by at most a constant factor, 
    for any particular vector $x_{I_j} \in [t]^{|I_j|}$. 
\end{itemize}

\subsection{A Lower Bound on {\EqualityTesting}}\label{sect:lb-et}

We begin with two general lemmas about discrete probability distributions that play an important role in our proof.

Roughly speaking, Lemma~\ref{lem:entropy-reduction-support} captures and generalizes the following intuition: 
Suppose $p$ is a \emph{high entropy} distribution on some universe $U$ and $q$ is obtained from $p$ by conditioning on an event $\X \subseteq U$ such that 
$p(\X)$ is large, say some constant like $1/4$. 
If $p$'s entropy is close to $\log|U|$, then $q$'s entropy should not be 
much smaller than that of $p$.
As our proof goes on round by round, we will constantly throw away part of the input distribution's support to meet certain conditions. It is Lemma~\ref{lem:entropy-reduction-support} that guarantees that the input distributions continue to have relatively high entropy.

Lemma~\ref{lem:entropy-to-ratio} comes into play because the error probability will be calculated backward in a round-by-round manner. Suppose the old distribution ($p$) has no extremely low probability 
point and the new distribution ($q$) has almost full entropy. Lemma~\ref{lem:entropy-to-ratio} provides us with a useful tool to transfer a lower bound on the probability of any event w.r.t.~$q$ to 
a lower bound on the same event w.r.t.~$p$.  It can be seen as a version of Markov's inequality for Kullback-Leibler divergences.

\begin{lemma}\label{lem:entropy-reduction-support}
Let $p$ and $q$ be distributions defined on a universe of size $2^s$. 
Suppose both of the following properties are satisfied:
\begin{enumerate2}
\item The entropy of $p$ is $\h{p} \ge s - g$, where $g \in [0,s)$;
\item There exists $\alpha \in (0,1)$ such that $q(x) \le p(x)/\alpha$ holds for every value $x \in \supp(q)$.
\end{enumerate2}
The entropy of $q$ is lower bounded by:
\[
\h{q} \ge s - g/\alpha - \h{\alpha}/\alpha.
\]
\end{lemma}
\begin{proof}
Let $\X$ be the whole universe. From our assumptions, the entropy of $q$ can be lower bounded as follows.
\begin{align*}
\h{q}
&= \sum_{x \in \X} q(x) \log \frac{1}{q(x)} & \mbox{Defn. of $\h{q}$.}\\
&= \frac{1}{\alpha}\sum_{x \in \X} \alpha q(x) \log \frac{1}{\alpha q(x)} + \log\alpha & \sum_{x \in \X} q(x) = 1.\\
&\ge \frac{1}{\alpha}\sum_{x \in \X} \left[p(x) \log \frac{1}{p(x)} - (p(x) - \alpha q(x)) \log \frac{1}{p(x) - \alpha q(x)}\right] + \log\alpha \\
\intertext{The previous step follows from Assumption 2
and the fact that $x\log x^{-1} + y\log y^{-1} \ge (x+y)\log (x+y)^{-1}$ for any $x,y \ge 0$. Continuing,}
&\ge  \frac{1}{\alpha}\left[s - g - \sum_{x \in \X} (p(x) - \alpha q(x)) \log \frac{1}{p(x) - \alpha q(x)}\right] + \log\alpha & \mbox{Assumption 1.}\\
&\ge \frac{1}{\alpha}\left[s - g - (1-\alpha)\log \frac{2^s}{1-\alpha}\right] + \log\alpha & \mbox{Concavity of logarithm.}\\
&=  s - \frac{g}{\alpha} + \frac{1-\alpha}{\alpha} \log (1-\alpha) + \log\alpha
\; = \; s - \frac{g}{\alpha} - \frac{\h{\alpha}}{\alpha}.
\end{align*}
\end{proof}

\begin{lemma}\label{lem:entropy-to-ratio}
Let $p$ and $q$ be distributions defined on a universe of size $2^s$. Suppose both of the following properties are satisfied:
\begin{enumerate2}
\item The entropy of $q$ is $\h{q} \ge s-g_1$, where $g_1 \in [0,s)$;
\item There exists $g_2 \ge 0$ such that $p(x) \ge 2^{-s-g_2}$ holds for every value $x \in \supp(q)$.
\end{enumerate2}
Then, for any $\alpha \in (0,1)$,
\[
\Pr_{x \sim q}\left[\frac{q(x)}{p(x)} > 2^{g_1/\alpha+g_2-(1-\alpha)\log(1-\alpha)/\alpha}\right] \le \alpha.
\]
\end{lemma}

\begin{remark}
Recall the Kullback-Leibler divergence (also known as relative entropy)
is defined to be $D_{\mathrm{KL}}(q \| p) = \sum_x q(x)\log\frac{q(x)}{p(x)}$, where $\supp(q)\subseteq \supp(p)$. I.e., it is the expected value of 
$\log\frac{q(x)}{p(x)}$ when $x\sim q$.  This lemma bounds the probability that $\log\frac{q(x)}{p(x)}$ deviates too far from 
its expectation.  It is syntactically similar to Markov's inequality,
but note that Markov's inequality is inapplicable as $\log\frac{q(x)}{p(x)}$ is generally not non-negative.
\end{remark}

\begin{proof}[Proof of Lemma~\ref{lem:entropy-to-ratio}]
Let $\X_0 = \{x \in \supp(q) \mid q(x)/p(x) \le 2^{g_1/\alpha+g_2-(1-\alpha)\log(1-\alpha)/\alpha}\}$ 
and $\X_1 = \supp(q) \setminus \X_0$. 
Suppose, for the purpose of obtaining a contradiction, 
that the conclusion of the lemma is false, i.e., 
$q(\X_1) = \alpha_0$, for some $\alpha_0 > \alpha$. 
Notice that for each value $x \in \X_1$, Assumption 2 implies that
\begin{equation}\label{eqn:qp}
q(x) > p(x)\cdot 2^{g_1/\alpha+g_2-(1-\alpha)\log(1-\alpha)/\alpha}  \ge 2^{-s+g_1/\alpha-(1-\alpha)\log(1-\alpha)/\alpha}.
\end{equation}
Then we can upper bound the entropy of $q$ as follows.
\begin{align*}
\h{q}
&=  \sum_{x \in \X_0} q(x)\log\frac{1}{q(x)} + \sum_{x \in \X_1} q(x)\log\frac{1}{q(x)} & \mbox{Defn. of $\h{q}$.}\\
&< \sum_{x \in \X_0} q(x)\log\frac{1}{q(x)} + \alpha_0\left[s-\frac{g_1}{\alpha}+\frac{1-\alpha}{\alpha}\log(1-\alpha)\right] & \mbox{Eqn.~(\ref{eqn:qp}).}\\
&\le (1-\alpha_0)\log\frac{2^s}{1-\alpha_0} + \alpha_0\left[s-\frac{g_1}{\alpha}+\frac{1-\alpha}{\alpha}\log(1-\alpha)\right] & \mbox{Concavity of logarithm.}\\
&= s - \frac{\alpha_0}{\alpha} \cdot g_1 + \alpha_0\left[\frac{1-\alpha}{\alpha}\log(1-\alpha) - \frac{1-\alpha_0}{\alpha_0}\log(1-\alpha_0)\right] \\
&< s - g_1,
\end{align*}
where the last step follows from the monotonicity of $(1-\alpha)\log(1-\alpha)/\alpha$. This contradicts Assumption 1.
\end{proof}

We are now ready to begin the proof of Theorem~\ref{thm:lb-et} proper. 
Fix a round $j$ and a particular history ($m_1,\ldots,m_{j-1}$) up to round $j-1$.  We let $\mu_j(m_j)$ denote the probability 
that the $j$th message is $m_j$, if the input to the sender 
is drawn from $\hD^{(j-1)}$.
Define $\D^{(j)}[m_j]$ to be the new input distribution 
of the sender after he commits to $m_j$.
When $m_j$ is clear from context, it is denoted $\D^{(j)}$.
(The process for deriving $\hD^{(j)}$ from $\D^{(j)}$ and $\D^{(j-1)}$
on the receiver's end will be explained in detail later.)

We will prove by induction that the following Invariant~\ref{inv:Ij}
holds for each $j\in [0,r]$, where the particular values
of $I_j$, $\D^{(j)}$, $\hD^{(j)}$, and $l_1,\ldots,l_j$
depend on the transcript $m_1,\ldots,m_j$ that is sampled.
In the base case, 
Invariant~\ref{inv:Ij} clearly holds when 
$j=0,I_0=[k]$, and both $\hD^{(0)},\D^{(0)}$ are the uniform
distribution over $[t]^k$.

\begin{invariant}\label{inv:Ij}
After round $j\in [0,r]$ the partial transcript is $m_1,\ldots,m_j$,
which determines the values 
$\{l_{j'}, \hD^{(j')}, \D^{(j')}, I_{j'}\}_{j'\leq j}$. 
The index set $I_j \subseteq [k]$ satisfies all of the following:
\begin{enumerate2}
\item $|I_j|{}= k^{1-j/r}$.
\item Each value $x_{I_j} \in [t]^{|I_j|}$ satisfies $\hD^{(j)}_{I_j}(x_{I_j}) \le 4\D^{(j)}_{I_j}(x_{I_j})$.
\item Each nonempty subset $I' \subseteq I_j$ satisfies
\[
\h{\hD^{(j)}_{I'}} \ge \left(cE - \sum_{u=1}^j \frac{\const^{j-u+1}l_u}{k^{1-(u-1)/r}} - \res^j\right) {}|I'|.
\]
\end{enumerate2}
\end{invariant}

In accordance with our informal discussion in Section~\ref{sect:desc-lb}, $I_j$ is a subset of indices on which both parties have learned little information about each other from the partial transcript $m_1,\ldots,m_j$.  
Invariant~\ref{inv:Ij}(2) ensures that the two parties draw their 
inputs after the $j$th round from similar distributions.
Invariant~\ref{inv:Ij}(3) is the most important property. 
It says that the information revealed by $\hD^{(j)}$ about $I'$ 
is roughly what one would expect, given the message lengths $l_1,\ldots,l_j$.
Note that the $u$th message conveys information about $|I_{u-1}|{}=k^{1-(u-1)/r}$ indices so the average information-per-index should be $l_u/k^{1-(u-1)/r}$.
The factor $\const^{j-u+1}$ and the extra 
term $\res^j$ come from Lemma~\ref{lem:entropy-reduction-support},
which throws away part of the input distribution in each round,
progressively distorting the distributions in minor ways.

To begin our induction, at round $j$ we find
a large fraction of possible messages $m_j$ 
that reveal little information about the sender's
input, projected onto $I_{j-1}$.
This is possible because the length of the message 
$l_j = {}|m_j|$ reflects an upper bound on the expected 
information gain.
This idea is formalized in the following Lemma~\ref{lem:good-message}.

\begin{lemma}\label{lem:good-message}
Fix $j\in [1,r]$ and suppose Invariant~\ref{inv:Ij} holds for $j-1$. Then there exists a subset of messages $\M'_j$ 
with $\mu_j(\M'_j) \ge 1/2$ such that each message 
$m_j \in \M'_j$ satisfies
\[
\h{\D^{(j)}_{I_{j-1}}[m_j]} \ge \left(cE-2\sum_{u=1}^j \frac{\const^{j-u}l_u}{k^{1-(u-1)/r}} - 2 \cdot \res^{j-1}\right){}|I_{j-1}|.
\]
\end{lemma}
\begin{proof}
Let $\M'_j$ contain all messages $m_j$ satisfying the above 
inequality and $\overline{\M'_j}$ be its complement. Suppose, for the purpose of obtaining a contradiction, that the conclusion of the lemma is not true, i.e., $\mu_j(\overline{\M'_j}) = \alpha > 1/2$. 
Then the entropy of $\hD^{(j-1)}_{I_{j-1}}$ can be upper bounded as follows.
\begin{align*}
\lefteqn{\h{\hD^{(j-1)}_{I_{j-1}}}}\\
&= \mi{\hD^{(j-1)}_{I_{j-1}}}{M_j} + \sum_{m_j \in (\M'_j \cup \overline{\M'_j})} \mu_j(m_j) \h{\D^{(j)}_{I_{j-1}}[m_j]} 
& \mbox{Defn.~of $I(\cdot,\cdot)$.}\\
&\le  \h{M_j} + \sum_{m_j \in (\M'_j \cup \overline{\M'_j})} \mu_j(m_j) \h{\D^{(j)}_{I_{j-1}}[m_j]} & \mi{X}{\cdot} \leq \h{X}.\\
&\le  l_j + \sum_{m_j \in \M'_j} \mu_j(m_j) \h{\D^{(j)}_{I_{j-1}}[m_j]} + \sum_{m_j \in \overline{\M'_j}} \mu_j(m_j) \h{\D^{(j)}_{I_{j-1}}[m_j]} & \mbox{$H(M_j) \leq {}|M_j|{} = l_j$.}\\
&< l_j + (1-\alpha) cE|I_{j-1}|{} + \alpha\left(cE-2\sum_{u=1}^j \frac{\const^{j-u}l_u}{k^{1-(u-1)/r}} - 2 \cdot \res^{j-1}\right)|I_{j-1}| & \mbox{Defn. of $\overline{\M'_j}$.}\\
&= l_j + \left(cE-2\alpha\sum_{u=1}^j \frac{\const^{j-u}l_u}{k^{1-(u-1)/r}} - 2\alpha \cdot \res^{j-1}\right)|I_{j-1}| \\
&< \left(cE-\sum_{u=1}^{j-1} \frac{\const^{j-u}l_u}{k^{1-(u-1)/r}} - \res^{j-1}\right)|I_{j-1}|, & \mbox{Because $\alpha > 1/2$.}
\end{align*}
This contradicts Invariant~\ref{inv:Ij}(3) at index $j-1$.
\end{proof}

After the $j$th message $m_j$ is sent, 
the next step is to identify a set of coordinates
$I_j$ such that $\D^{(j)}$ still reveals little information
about $I_j$ \emph{and every subset of $I_j$}, since we need this
property to hold for $I_{j+1},\ldots,I_r$ in the future, all
of which are subsets of $I_j$.  We also want $I_j$ not to contain
many low probability points w.r.t.~$\D^{(j-1)}$, since this
may stop us from applying Lemma~\ref{lem:entropy-to-ratio} later on.
These two constraints are captured by parts (2) and (1),
respectively, of Lemma~\ref{lem:good-index}.

\begin{lemma}\label{lem:good-index}
Fix $j\in[1,r]$ and suppose Invariant~\ref{inv:Ij} holds for $j-1$. Then there exists a subset of messages $\M_j \subseteq \M'_j$ 
(from Lemma~\ref{lem:good-message}) with $\mu_j(\M_j) \ge 1/4$ such that for each message $m_j \in \M_j$, there exists a subset $I_j \subseteq I_{j-1}$ of size $|I_j|{}=k^{1-j/r}$ satisfying both of the following properties:
\begin{enumerate2}
\item $\Pr_{x_{I_j} \sim \D^{(j)}_{I_j}}\left[\D^{(j-1)}_{I_j}(x_{I_j}) < (4t)^{-{}|I_j|}/32\right] \le 1/2$;
\item Each nonempty subset $I' \subseteq I_j$ satisfies
\[
\h{\D^{(j)}_{I'}} \ge \left(cE - 4\sum_{u=1}^j \frac{\const^{j-u}l_u}{k^{1-(u-1)/r}} - 4 \cdot \res^{j-1}\right)|I'|.
\]
\end{enumerate2}
\end{lemma}
\begin{proof}
We first prove that for each message $m_j \in \M'_j$ (from Lemma~\ref{lem:good-message}), there exists a subset $J_0 \subseteq I_{j-1}$ of size $|J_0|{} \geq {}|I_{j-1}|/2$ such that each nonempty subset 
$I' \subseteq J_0$ satisfies part 
(2) of the lemma. 
Suppose $J_1,J_2,\ldots,J_w$ are disjoint subsets of $I_{j-1}$, each of which \emph{violates} the inequality of part (2), 
whereas none of the subsets of $J_0 = I_{j-1} \setminus (\bigcup_{v=1}^w J_v)$ do. 
Then we can upper bound the entropy of $\D^{(j)}_{I_{j-1}}$ as follows.
\begin{align*}
\h{\D^{(j)}_{I_{j-1}}}
&\le \sum_{v=0}^w \h{\D^{(j)}_{J_v}} & \mbox{Fact~\ref{fac:joint-entropy}.}\\
&< cE|J_0| {} + \sum_{v=1}^w \left(cE - 4\sum_{u=1}^j\frac{\const^{j-u}l_u}{k^{1-(u-1)/r}} - 4 \cdot \res^{j-1}\right)|J_v| & \mbox{Defn. of $J_v$.}\\
&= cE|I_{j-1}| {} - 4|I_{j-1} \setminus J_0| \left(\sum_{u=1}^j\frac{\const^{j-u}l_u}{k^{1-(u-1)/r}} + \res^{j-1}\right).
\intertext{On the other hand, from Lemma~\ref{lem:good-message}, having $m_j \in \M'_j$ guarantees that}
\h{\D^{(j)}_{I_{j-1}}} 
&\ge \left(cE - 2\sum_{u=1}^j \frac{\const^{j-u}l_u}{k^{1-(u-1)/r}} - 2 \cdot \res^{j-1}\right)|I_{j-1}|.
\end{align*}
The two inequalities above are only consistent if
$|I_{j-1} \setminus J_0| {}\le{} |I_{j-1}|/2$, or equivalently 
$|J_0| {}\ge{} |I_{j-1}|/2$. 
Thus, $J_0$ exists with the right cardinality, as claimed.

Now suppose, for the purpose of obtaining a contradiction,
that the lemma is false.  For every $m_j \in \M'_j$ there is
a corresponding index set $J_0$ whose subsets satisfy part (2)
of the lemma.  If the lemma is false, 
that means there is a subset $\M''_j \subseteq \M'_j$
of ``bad'' messages with $\mu_j(\M''_j) > 1/4$
such that, for each $m_j\in \M''_j$, none of the $\binom{|J_0|}{|I_j|}$ choices for $I_j \subseteq J_0$ satisfy part (1) of the lemma.
(Remember that $J_0$ depends on $m_j$ but the lower bound on $|J_0|{}\geq{}|I_{j-1}|/2$ is independent of $m_j$.)
Consider the following summation:
\[
Z = \sum_{\substack{I_j \subseteq I_{j-1}\:: \\ |I_j|{}=k^{1-j/r}}} \sum_{\substack{x_{I_j} \in [t]^{|I_j|}\:: \\ \D^{(j-1)}_{I_j}(x_{I_j}) \,<\, (4t)^{-{}|I_j|}/32}} \D^{(j-1)}_{I_j}(x_{I_j}).
\]
We can easily upper bound $Z$ as follows.
\[
Z < \binom{|I_{j-1}|}{|I_j|} \cdot t^{|I_j|} \cdot \frac{(4t)^{-{}|I_j|}}{32} = \binom{|I_{j-1}|}{|I_j|} 2^{-2|I_j|{}-5}.
\]
Invariant~\ref{inv:Ij}(2) relates $\D^{(j-1)}$ and $\hD^{(j-1)}$,
which lets us lower bound $Z$.
\begin{align*}
Z
&\ge \frac{1}{4}\sum_{\substack{I_j \subseteq I_{j-1}\::\\ |I_j|{}=k^{1-j/r}}} \sum_{\substack{x_{I_j} \in [t]^{|I_j|}\:: \\ \D^{(j-1)}_{I_j}(x_{I_j}) \,<\, (4t)^{-{}|I_j|}/32}} \hD^{(j-1)}_{I_j}(x_{I_j}) & \mbox{Invariant~\ref{inv:Ij}(2).}\\
\intertext{By definition, $\hD^{(j-1)}$ is a convex combination
of the $\D^{(j)}[m_j]$ distributions, weighted according to $\mu_j(\cdot)$.
Hence, the expression above is lower bounded by}
&\ge \frac{1}{4}\sum_{\substack{I_j \subseteq I_{j-1}\:: \\ |I_j|{}=k^{1-j/r}}} \sum_{\substack{x_{I_j} \in [t]^{|I_j|}\:: \\ \D^{(j-1)}_{I_j}(x_{I_j}) \,<\, (4t)^{-{}|I_j|}/32}} \sum_{m_j \in \M''_j} \mu_j(m_j) \cdot \D^{(j)}_{I_j}[m_j](x_{I_j}) \\
&\ge \frac{1}{4}\sum_{m_j \in \M''_j} \mu_j(m_j) \sum_{\substack{I_j \subseteq J_0\:: \\ |I_j|{}=k^{1-j/r}}} \sum_{\substack{x_{I_j} \in [t]^{|I_j|}\:: \\ \D^{(j-1)}_{I_j}(x_{I_j}) \,<\, (4t)^{-{}|I_j|}/32}} \D^{(j)}_{I_j}[m_j](x_{I_j}) & \mbox{Rearrange sums.}\\
\intertext{By definition, for every $m_j\in \M''_j$ and 
every choice of $I_j\subseteq J_0$, part (1) of the lemma is violated.
Continuing with the inequalities,}
&> \frac{1}{4}\sum_{m_j \in \M''_j} \mu_j(m_j) \cdot \binom{|J_0|}{|I_j|} \cdot \frac{1}{2} \\
&> \frac{1}{32}\binom{|I_{j-1}|/2}{|I_j|}. & \mbox{Because $\mu_j(\M''_j) > 1/4$.}
\end{align*}
This contradicts the upper bound on $Z$ whenever $k^{1/r}$ is at least some sufficiently large constant.
\end{proof}

The receiver of $m_j$ constructs a new distribution $\hD^{(j)}$
in two steps.  After fixing $I_j$, 
we construct $\tD^{(j)}$ by combining $\D^{(j-1)}$ and $\D^{(j)}$,
filtering out some points in the space whose probability mass is too low.  
We then construct $\hD^{(j)}$ from $\tD^{(j)}$ and $\D^{(j-1)}$
by filtering out points that occur under $\tD^{(j)}$ with substantially larger probability than they do under $\D^{(j-1)}$.

Formally, suppose Invariant~\ref{inv:Ij} holds for $j-1$. 
For each message $m_j \in \M_j$ (from Lemma~\ref{lem:good-index}), 
let $I_j$ be selected to satisfy both 
properties of Lemma~\ref{lem:good-index}. 
Define the probability mass of a vector $x \in [t]^k$ under $\tD^{(j)}$ as follows:
\begin{align*}
\tD^{(j)}(x) 
&=
\begin{cases}
0, & \mbox{if } \D^{(j-1)}_{I_j}(x_{I_j}) < \frac{(4t)^{-{}|I_j|}}{32}; \\
\frac{\D^{(j)}_{I_j}(x_{I_j})}{\beta_1} \cdot \frac{\D^{(j-1)}(x)}{\D^{(j-1)}_{I_j}(x_{I_j})}, & \text{otherwise}.
\end{cases}
\intertext{where $\beta_1$ is}
\beta_1 
&= \Pr_{x_{I_j} \sim \D^{(j)}_{I_j}}\left[\D^{(j-1)}_{I_j}(x_{I_j}) \ge \frac{(4t)^{-{}|I_j|}}{32}\right].
\intertext{In other words, we discard a $1-\beta_1$ fraction
of the distribution $\D^{(j)}$, but ignoring this effect,
the projection of $\tD^{(j)}$ onto $I_j$ 
has the same distribution as $\D^{(j)}$ onto $I_j$, and conditioned
on the value of $x_{I_j}$, the distribution 
$\tD^{(j)}$ (projected onto $[k]\backslash I_j$) 
is identical to $\D^{(j-1)}$.  We derive $\hD^{(j)}$ 
from $\tD^{(j)}$ with a similar transformation.}
\hD^{(j)}(x) &= 
\begin{cases}
0, & \mbox{if } \frac{\tD^{(j)}_{I_j}(x_{I_j})}{\D^{(j-1)}_{I_j}(x_{I_j})} > 2^{\gamma_j}; \\
\frac{\tD^{(j)}_{I_j}(x_{I_j})}{\beta_2} \cdot \frac{\D^{(j-1)}(x)}{\D^{(j-1)}_{I_j}(x_{I_j})}, & \text{otherwise}.
\end{cases}
\intertext{where $\beta_2$ and $\gamma_j$ are defined to be}
\beta_2 
&= \Pr_{x_{I_j} \sim \tD^{(j)}_{I_j}}\left[\frac{\tD^{(j)}_{I_j}(x_{I_j})}{\D^{(j-1)}_{I_j}(x_{I_j})} \le 2^{\gamma_j}\right],\\
\gamma_j
&= \left(\sum_{u=1}^j l_u\left(\frac{\const^{j-u+1}}{k^{1 - (u-1)/r}}\right) + (16 \cdot \res^{j-1}+6)\right)|I_j| {} + 6\\
&\le \sum_{u=1}^j l_u\left(\frac{\const}{k^{1/r}}\right)^{j-u+1} + \res^j\cdot |I_j| {} + 6.
\end{align*}

The proofs of Lemmas~\ref{lem:entropy-reduction} and \ref{lem:err-prob-reduction} use 
several simple observations about $\tD^{(j)}$ and $\hD^{(j)}$:
\begin{enumerate2}
\item Lemma~\ref{lem:good-index}(1) states that $\beta_1 \ge 1/2$.
Lemma~\ref{lem:good-index}(2) lower bounds the entropy of $\D^{(j)}_{I_j}$.
We apply Lemma~\ref{lem:entropy-reduction-support} to 
$\D^{(j)}_{I_j}$ and $\tD^{(j)}_{I_j}$ 
(taking the roles of $p$ and $q$, respectively)
with parameter $\alpha=1/2\leq \beta_1$, and obtain the following lower bound on the entropy
of $\tD^{(j)}_{I_j}$.
\[
\h{\tD^{(j)}_{I_j}} \ge \left(cE - 8\sum_{u=1}^j \frac{\const^{j-u}l_u}{k^{1-(u-1)/r}} - 8 \cdot \res^{j-1} - 2\right) |I_j|.
\]
\item We can then apply Lemma~\ref{lem:entropy-to-ratio} to $\D^{(j-1)}_{I_j}$ and $\tD^{(j)}_{I_j}$ (taking the roles of $p$ and $q$, respectively) with parameters
\begin{align*}
g_1 &= \left(8\sum_{u=1}^j \frac{\const^{j-u}l_u}{k^{1-(u-1)/r}} + (8 \cdot \res^{j-1} + 2)\right)|I_j|,\\
g_2 &= 2|I_j| {} + 5,\\
\mbox{ and } \alpha &=1/2.
\end{align*}
Since $g_1/\alpha + g_2 - (1-\alpha)\log(1-\alpha)/\alpha = \gamma_j$,
we conclude that $\beta_2 \geq 1-\alpha = 1/2$.
Thus, for each value $x_{I_j} \in \supp(\hD^{(j)}_{I_j})$, 
\begin{equation}\label{eqn:timesfour}
\hD^{(j)}_{I_j}(x_{I_j}) = \frac{\tD^{(j)}_{I_j}(x_{I_j})}{\beta_2} = \frac{\D^{(j)}_{I_j}(x_{I_j})}{\beta_1\beta_2} \le 4\D^{(j)}_{I_j}(x_{I_j}).
\end{equation}
\end{enumerate2}

Lemma~\ref{lem:entropy-reduction} completes the inductive step 
by lower bounding the entropy of $\hD^{(j)}_{I'}$ for every nonempty subset $I' \subseteq I_j$. To put it another way, it ensures that the coordinates in $I_j$ remain almost completely unknown to both parties.

\begin{lemma}\label{lem:entropy-reduction}
Fix $j\in [1,r]$ and suppose Invariant~\ref{inv:Ij} holds for $j-1$. Then, for each message $m_j \in \M_j$ (from Lemma~\ref{lem:good-index}), Invariant~\ref{inv:Ij} also holds for $j$.
\end{lemma}
\begin{proof}
Due to Lemma~\ref{lem:good-index} and Eqn.~(\ref{eqn:timesfour}), the first two properties of Invariant~\ref{inv:Ij} are satisfied. For each nonempty subset $I' \subseteq I_j$, the third property of Invariant~\ref{inv:Ij} can be derived from the second property of Lemma~\ref{lem:good-index} and an application of Lemma~\ref{lem:entropy-reduction-support} to $\D^{(j)}_{I'}$ and $\hD^{(j)}_{I'}$ (taking the roles of $p$ and $q$, respectively) with parameter $\alpha=1/4$ as follows.
\[
\h{\hD^{(j)}_{I'}} \ge \left(cE - 16\sum_{u=1}^j \frac{\const^{j-u}l_u}{k^{1-(u-1)/r}} - 16 \cdot \res^{j-1} - 4\right) |I'| {}
\;\ge\; 
\left(cE - \sum_{u=1}^j \frac{\const^{j-u+1}l_u}{k^{1-(u-1)/r}} - \res^j\right) |I'|.
\]
\end{proof}

Aside from maintaining Invariant~\ref{inv:Ij} round by round, another important part of our proof is to compute the error probability. Lemma~\ref{lem:err-prob-reduction} shows how the error probabilities of two consecutive rounds are related after our modification to the protocol. More importantly, it also illustrates the reason to bound the \emph{pointwise} ratio 
between $\tD^{(j)}_{I_j}$ and $\D^{(j-1)}_{I_j}$.

\begin{lemma}\label{lem:err-prob-reduction}
Fix a round $j\in[1,r]$ and suppose Invariant~\ref{inv:Ij} holds for $j-1$. Fix any specific message $m_j \in \M_j$ (from Lemma~\ref{lem:good-index}).
Define $p$ to be the probability of error, when the protocol
begins after round $j$ with the inputs drawn from
$\D^{(j)}$ and $\hD^{(j)}$, respectively. 
Then the probability of error is at least 
$2^{-\gamma_j-1}p$ 
when the inputs are instead drawn from
$\D^{(j)}$ and $\D^{(j-1)}$, respectively.
\end{lemma}
\begin{proof}
From the definition of $\hD^{(j)}$, 
for each value $x \in \supp(\hD^{(j)})$, we have
\begin{equation}\label{eqn:ratio}
\frac{\hD^{(j)}(x)}{\D^{(j-1)}(x)} 
\:=\: \frac{\tD^{(j)}_{I_j}(x_{I_j})}{\beta_2\D^{(j-1)}_{I_j}(x_{I_j})} 
\:\le\: \frac{2^{\gamma_j}}{\beta_2} \;\leq\; 2^{\gamma_j+1}.
\end{equation}
This concludes the proof.
\end{proof}

Finally, with all lemmas proved above, we have reached the point to calculate the initial error probability.

\begin{lemma}\label{lem:real-err-prob-et}
Recall that $c=1/2, c'=c/100$.  Fix any $r\in [1,(\log k)/6]$
and $E \ge 100k^{1-1/r}/c$.
Suppose the initial input vectors are drawn independently and uniformly from 
$[t]^k$, where $t=2^{cE}$. 
Then the error probability of the {\EqualityTesting} protocol, $\perr$,
is greater than $2^{-E}$.
\end{lemma}
\begin{proof}
First suppose Invariant~\ref{inv:Ij} holds for $r$ and consider the situation after the final round, where the inputs are drawn from $\D^{(r)}$ and $\hD^{(r)}$, respectively. Notice that $I_r$ is a singleton set, so the entropy of $\hD^{(r)}_{I_r}$ can be lower bounded as follows.
\begin{align*}
\h{\hD^{(r)}_{I_r}}
&\ge cE - \sum_{u=1}^r \frac{\const^{r-u+1}l_u}{k^{1-(u-1)/r}} - \res^r & \mbox{Invariant~\ref{inv:Ij}(3).}\\
&= cE - \frac{\const}{k^{1/r}} \sum_{u=1}^r l_u\left(\frac{\const}{k^{1/r}}\right)^{r-u} - \res^r \\
&\ge cE - \frac{\const}{k^{1/r}} \sum_{u=1}^r l_u - \res k^{1-1/r} & \mbox{$k^{1/r} \ge 2^6$ due to $r\le (\log k)/6$.}\\
&\ge cE - \const c'E - \res k^{1-1/r} \:>\:\frac{cE}{2}. & \mbox{Because $\sum_{u=1}^r l_u \le c'Ek^{1/r}$.}
\end{align*}

From the lower bound on the entropy of $\hD^{(r)}_{I_r}$, 
we can easily show that there exists no value $x_{I_r}$ such that $\hD^{(r)}_{I_r}(x_{I_r}) = \alpha > 3/4$. 
If there \emph{were} such a value, then 
the entropy of 
$\hD^{(r)}_{I_r}$ can also be upper bounded as
\[
\h{\hD^{(r)}_{I_r}} \;\le\; \alpha \log \frac{1}{\alpha} + (1-\alpha) \log \frac{t}{1-\alpha} \;<\;\frac{cE}{4} + \alpha \log \frac{1}{\alpha} + (1-\alpha) \log \frac{1}{1-\alpha} \;<\; \frac{cE}{2},
\]
contradicting the lower bound on $\h{\hD^{(r)}_{I_r}}$.

After all $r$ rounds of communication, 
the receiver of the last message has to make 
the decision on $I_r$ depending only on his own input on $I_r$. 
Let $\X_0 \subseteq [t]$ be the subset of values $x_{I_r}$ such that the protocol outputs ``not equal'' on $I_r$ upon seeing the input $x_{I_r}$ after $r$ rounds of communication, $\X_1 = [t] \setminus \X_0$, and $\beta = \hD^{(r)}_{I_r}(\X_0)$. Then, the final error probability is at least
\begin{align*}
\lefteqn{\sum_{x_{I_r} \in \X_0} \hD^{(r)}_{I_r}(x_{I_r}) \D^{(r)}_{I_r}(x_{I_r}) + \sum_{x_{I_r} \in \X_1} \hD^{(r)}_{I_r}(x_{I_r}) \left(1 - \D^{(r)}_{I_r}(x_{I_r})\right)} \\
&= \sum_{x_{I_r} \in \X_0} \hD^{(r)}_{I_r}(x_{I_r}) \D^{(r)}_{I_r}(x_{I_r}) + \sum_{x_{I_r} \in \X_1} \hD^{(r)}_{I_r}(x_{I_r}) \sum_{x'_{I_r} \neq x_{I_r}} \D^{(r)}_{I_r}(x'_{I_r}) \\
&\ge \frac{1}{4}\sum_{x_{I_r} \in \X_0} \hD^{(r)}_{I_r}(x_{I_r})^2 + \frac{1}{4}\sum_{x_{I_r} \in \X_1} \hD^{(r)}_{I_r}(x_{I_r}) \sum_{x'_{I_r} \neq x_{I_r}} \hD^{(r)}_{I_r}(x'_{I_r}) & \mbox{Invariant~\ref{inv:Ij}(2).}\\
&= \frac{1}{4}\sum_{x_{I_r} \in \X_0} \hD^{(r)}_{I_r}(x_{I_r})^2 + \frac{1}{4}\sum_{x_{I_r} \in \X_1} \hD^{(r)}_{I_r}(x_{I_r}) \left(1 - \hD^{(r)}_{I_r}(x_{I_r})\right) & \\
&\ge \frac{1}{4}\sum_{x_{I_r} \in \X_0} \hD^{(r)}_{I_r}(x_{I_r})^2 + \frac{1}{16}\sum_{x_{I_r} \in \X_1} \hD^{(r)}_{I_r}(x_{I_r}) & \mbox{Because $\hD^{(r)}_{I_r}(x_{I_r}) \le 3/4$.}\\
&\ge \frac{\beta^2}{4t} + \frac{1-\beta}{16} \:\ge\: \frac{1}{4t}. & \mbox{Convexity of $x^2$.}
\end{align*}

This result also meets the simple intuition that when the inputs to the two parties are almost uniformly random and no communication is allowed, the best strategy would be guessing ``not equal'' regardless of the actual input.

Finally, we are ready to transfer the error probability back round by round. From Lemma~\ref{lem:good-index} through Lemma~\ref{lem:err-prob-reduction}, 
the error probability w.r.t.~$\D^{(j)}$ and $\hD^{(j)}$ 
differs from the error probability w.r.t.~$\D^{(j-1)}$ and $\hD^{(j-1)}$ by at most a $4 \cdot 2^{\gamma_j+1} = 2^{\gamma_j+3}$ factor.
In particular, Lemma~\ref{lem:good-index} and Lemma~\ref{lem:entropy-reduction} say that the 
$j$th message $m_j$ satisfies Invariant~\ref{inv:Ij} at index $j$
with probability at least $1/4$,
provided Invariant~\ref{inv:Ij} holds for $j-1$,
and Lemma~\ref{lem:err-prob-reduction} says 
the error probabilities under the two measures differ by a $2^{\gamma_j+1}$ factor for any such $m_j$. 
Repeating this for each $j\in [1,r]$, 
we conclude that the initial error probability $\perr$ 
is lower bounded by
\begin{align*}
\perr &\ge \zero{\frac{1}{4t}\cdot \exp\left(-3r-\sum_{j=1}^r\gamma_j\right) 
= \exp\left(-cE-2-3r-\sum_{j=1}^r\gamma_j\right) > 2^{-E},}
\intertext{since}
\lefteqn{cE + 2 + 3r + \sum_{j=1}^r \gamma_j}\\
&\le cE + 2 + 3r + 6r + \sum_{j=1}^r \sum_{u=1}^j l_u \left(\frac{\const}{k^{1/r}}\right)^{j-u+1} \zero{+ \sum_{j=1}^r \res^j |I_j|}\\
&\le cE + 11r + \sum_{u=1}^r \frac{\const l_u}{k^{1/r}} \sum_{j=u}^r \left(\frac{\const}{k^{1/r}}\right)^{j-u} \zero{+ \res k^{1-1/r} \sum_{j=1}^r \left(\frac{\res}{k^{1/r}}\right)^{j-1}}\hspace*{3cm} & \mbox{Rearrange sums.}\\
&\le cE + 11r + \frac{32}{k^{1/r}} \sum_{u=1}^r l_u + 44k^{1-1/r} & \mbox{$k^{1/r} \ge 2^6$ since $r\le (\log k)/6$.}\\
&\le cE + \frac{11cE}{100} + \frac{32cE}{100} + \frac{44cE}{100} \:<\: E. & \mbox{Because $\sum_{u=1}^r l_u \le c'Ek^{1/r}$.}
\end{align*}
\end{proof}

\begin{proof}[Proof of Theorem~\ref{thm:lb-et}]
Lemma~\ref{lem:real-err-prob-et} actually shows that given integers $k \ge 1$ and $r \le (\log k)/6$, any $r$-round deterministic protocol for {\EqualityTesting} on vectors of length $k$ that has distributional error probability $\perr = 2^{-E}$ with respect to the uniform 
input distribution on $[t]^k$, where $t=2^{cE}$,
requires at least $\Omega(Ek^{1/r})$ bits of communication. 
Notice that the additional assumption $E \ge 100k^{1-1/r}/c$ 
always makes sense since there is a trivial $\Omega(k)$
lower bound on the communication complexity of 
$\EqualityTesting$, regardless of $r$.
Thus, Theorem~\ref{thm:lb-et} follows directly 
from Yao's minimax principle.
\end{proof}

\subsection{A Lower Bound on {\ExistsEqual}}\label{sect:lb-ee}

The proof of Theorem~\ref{thm:lb-ee} is almost the same as that of Theorem~\ref{thm:lb-et}, except for the final step, namely Lemma~\ref{lem:real-err-prob-et}, in which we first compute the final error probability after all $r$ rounds of communication and then transfer it backward round by round using Lemma~\ref{lem:err-prob-reduction}.
The problem with applying the same argument to {\ExistsEqual} protocols is that the receiver of the last message may be able to announce the correct answer, even though it knows little information about the 
inputs on the single coordinate $I_r$.

In order to prove Theorem~\ref{thm:lb-ee}, first notice that Lemma~\ref{lem:good-message} through Lemma~\ref{lem:err-prob-reduction} also hold perfectly well for {\ExistsEqual} protocols as no modification is required in their proofs. Therefore, it is sufficient to prove the following Lemma~\ref{lem:real-err-prob-ee}, which is an analog of Lemma~\ref{lem:real-err-prob-et} for {\ExistsEqual}. It is based mainly on Markov's inequality.

\begin{lemma}\label{lem:real-err-prob-ee}
Recall that $c=1/2, c'=c/100$.
Consider an execution of a deterministic 
$r$-round {\ExistsEqual} protocol, $r\in [1,(\log k)/6]$,
on input vectors drawn independently and uniformly
from $[t]^k$, where $t=2^{cE}$.
Here $E \ge 100k^{1-1/r}/c$ if $r > 1$ and $E \ge (100\log k)/c$ otherwise.
Then the protocol errs with probability $\perr > 2^{-E}$.
\end{lemma}
\begin{proof}
Similarly to the proof of Lemma~\ref{lem:real-err-prob-et}, we first consider the situation after the final round.
In the {\ExistsEqual} protocol, the receiver of the last message can make the decision depending on every coordinate of his own input.
Let $\X_0 \subseteq [t]^k$ be the subset of values $x$ such that the protocol outputs ``no'' upon seeing the input $x$ after $r$ rounds of communication, $\X_1 = [t]^k \setminus \X_0$.
Then, the final error probability is at least
\[
\sum_{x \in \X_0} \hD^{(r)}(x) \D^{(r)}_{I_r}(x_{I_r}) + \sum_{x \in \X_1} \hD^{(r)}(x) \left(1 - \sum_{y \in \N(x)} \D^{(r)}(y)\right),
\]
where $\N(x) = \{y \in [t]^k \mid \text{there exists some $i \in [k]$ such that $x_i = y_i$}\}$ is the subset of input vectors 
that agree with $x$ on at least one coordinate.

The main difficulty here is to lower bound $1 - \sum_{y \in \N(x)} \D^{(r)}(y)$, which is potentially quite small.
Consider the following summation $Z_0$ over \emph{all} transcripts $m_1,\ldots,m_r$ in which $m_j\in \M_j$ (from Lemma~\ref{lem:good-index}), where the set $\M_j$
depends on $m_1,\ldots,m_{j-1}$:
\[
Z_0 = \sum_{m_1 \in \M_1} \mu_1(m_1) \sum_{m_2 \in \M_2} \mu_2(m_2) \cdots \sum_{m_r \in \M_r} \mu_r(m_r) \sum_{x \in [t]^k} \hD^{(r)}(x) \sum_{y \in \N(x)} \D^{(r)}(y).
\]
From the proof of Lemma~\ref{lem:err-prob-reduction} (Eqn.~(\ref{eqn:ratio})), 
we can upper bound $Z_0$ as follows.
\begin{align*}
Z_0
& \le \sum_{m_1 \in \M_1} \mu_1(m_1) \cdots \sum_{m_r \in \M_r} \mu_r(m_r) \sum_{\substack{x \in [t]^k, \\ y \in \N(x)}} 2^{\gamma_r+1}\cdot \zero{\D^{(r-1)}(x) \cdot \D^{(r)}(y)}\\
\intertext{%
Notice that $\gamma_r$ and $\D^{(r-1)}$ are independent of the choice of $m_r$, hence by rearranging sums, this is equal to
}
& = \sum_{m_1 \in \M_1} \mu_1(m_1) \cdots \sum_{m_{r-1} \in \M_{r-1}} \mu_{r-1}(m_{r-1}) \sum_{\substack{x \in [t]^k, \\ y \in \N(x)}} 2^{\gamma_r+1}\cdot \D^{(r-1)}(x) \sum_{m_r \in \M_r} \mu_r(m_r)\cdot \D^{(r)}(y) \\
\intertext{%
By definition, $\hD^{(r-1)}$ is a convex combination
of the $\D^{(r)}[m_r]$ distributions, weighted according to $\mu_r(\cdot)$.
Hence, the expression above is upper bounded by
}
& \le \sum_{m_1 \in \M_1} \mu_1(m_1) \cdots \sum_{m_{r-1} \in \M_{r-1}} \zero{\mu_{r-1}(m_{r-1}) \sum_{\substack{x \in [t]^k, \\ y \in \N(x)}} 2^{\gamma_r+1} \cdot \D^{(r-1)}(x) \cdot \hD^{(r-1)}(y)} 
\intertext{%
By the symmetry of $x$ and $y$, this is equal to
}
& = \sum_{m_1 \in \M_1} \mu_1(m_1) \cdots \sum_{m_{r-1} \in \M_{r-1}} \zero{\mu_{r-1}(m_{r-1}) \sum_{\substack{x \in [t]^k, \\ y \in \N(x)}} 2^{\gamma_r+1} \cdot \hD^{(r-1)}(x)\cdot  \D^{(r-1)}(y)}
\intertext{We repeat the same argument for rounds $r-1$ down to $1$, upper bounding $Z_0$ by}
& \le \exp\left(r+\sum_{j=1}^r \gamma_j\right)\cdot \sum_{\substack{x \in [t]^k, \\ y \in N(x)}} \hD^{(0)}(x) \cdot \D^{(0)}(y) & \\
& \le \exp\left(r+\sum_{j=1}^r \gamma_j\right)\cdot \frac{k}{t}
\intertext{The last inequality above follows from a union bound
since, under the initial distributions $\hD^{(0)},\D^{(0)}$, 
each of the $k$ coordinates is equal with probability $1/t$.
Recall that $E \ge 100k^{1-1/r}/c$ when $r > 1$ and $E \ge (100\log k)/c$ otherwise. Hence, using the same argument as that in the proof of Lemma~\ref{lem:real-err-prob-et}, we can further bound this as}
& \le 2^{0.83cE} \cdot 2^{0.02cE} \cdot 2^{-cE} \:=\: 2^{-0.15cE},
\end{align*}
since
\[
r + \sum_{j=1}^r \gamma_j \le 7r + \sum_{j=1}^r \sum_{u=1}^j l_u \left(\frac{\const}{k^{1/r}}\right)^{j-u+1} + \sum_{j=1}^r \res^j |I_j| {} \le \frac{7cE}{100} + \frac{32cE}{100} + \frac{44cE}{100} = \frac{83cE}{100},
\]
and $k \le (cE/100)^{r/(r-1)} \le (cE/100)^2 \le 2^{0.02cE}$ when $r > 1$ and $k \le 2^{0.01cE}$ otherwise.

Now fix a round $j$ and a particular history $(m_1,\ldots,m_j)$ up to round $j$ such that $m_{j'} \in \M_{j'}$ holds for every $j' \le j$.
Define $Z_j$ as follows.
\[
Z_j = \sum_{m_{j+1} \in \M_{j+1}} \mu_{j+1}(m_{j+1}) \cdots \sum_{m_r \in \M_r} \mu_r(m_r) \sum_{x \in [t]^k} \hD^{(r)}(x) \sum_{y \in \N(x)} \D^{(r)}(y).
\]
By Markov's inequality, there exists a subset of messages $\hM_1 \subseteq \M_1$ with $\mu_1(\hM_1) \ge \mu_1(\M_1)/2 \ge 1/8$ such that each message $m_1 \in \hM_1$ satisfies $Z_1 \le 2Z_0/\mu_1(\M_1) \le 8Z_0$ since $\mu_1(\M_1) \ge 1/4$ from Lemma~\ref{lem:good-index}.
Similarly, conditioned on any specific $m_1 \in \hM_1$, by Markov's inequality, there exists a subset of messages $\hM_2 \subseteq \M_2$ with $\mu_2(\hM_2) \ge \mu_2(\M_2)/2 \ge 1/8$ such that each message $m_2 \in \hM_2$ satisfies $Z_2 \le 2Z_1/\mu_2(\M_2) \le 8^2Z_0$.
In general, conditioned on any specific partial transcript $m_1,\ldots,m_{j-1}$ such that $m_{j'} \in \hM_{j'}$ holds for every $j' < j$, there exists a subset of messages $\hM_j \subseteq \M_j$ with $\mu_j(\hM_j) \ge \mu_j(\M_j)/2 \ge 1/8$ such that each message $m_j \in \hM_j$ satisfies $Z_j \le 8^jZ_j$.

After repeating the same argument $r$ times, we get $\hM_1,\ldots,\hM_r$ in sequence. For any sampled transcript $m_1,\ldots,m_r$ such that $m_j \in \hM_j$ for all $j \le r$, we have
\[
Z_r \le 8^rZ_0 \le 2^{3r} \cdot 2^{-0.15cE} \le 2^{-0.12cE} \le \frac{1}{4},
\]
as $r \le cE/100$ and $cE \ge 100$.
Further, one more application of Markov's inequality shows that there exists a subset of values $\X' \subseteq [t]^k$ with $\hD^{(r)}(\X') = \alpha \ge 1/2$ such that $\sum_{y \in \N(x)} \D^{(r)}(y) \le 1/2$ holds for every $x \in \X'$.

As a result, we can then lower bound the final error probability as follows, where $\beta = \hD^{(r)}(\X_0 \cap \X')$.
\begin{align*}
\lefteqn{\sum_{x \in \X_0} \hD^{(r)}(x) \D^{(r)}_{I_r}(x_{I_r}) + \sum_{x \in \X_1} \hD^{(r)}(x) \left(1 - \sum_{y \in \N(x)} \D^{(r)}(y)\right)} \\
& \ge \sum_{x \in (\X_0 \cap \X')} \hD^{(r)}(x) \D^{(r)}_{I_r}(x_{I_r}) + \sum_{x \in (\X_1 \cap \X')} \hD^{(r)}(x) \left(1 - \sum_{y \in \N(x)} \D^{(r)}(y)\right) \\
& \ge \frac{1}{4}\sum_{x \in (\X_0 \cap \X')} \hD^{(r)}(x) \hD^{(r)}_{I_r}(x_{I_r}) + \sum_{x \in (\X_1 \cap \X')} \hD^{(r)}(x) \left(1 - \sum_{y \in \N(x)} \D^{(r)}(y)\right) & \mbox{Invariant~\ref{inv:Ij}(2).}\\
& \ge \frac{1}{4}\sum_{x \in (\X_0 \cap \X')} \hD^{(r)}(x) \hD^{(r)}_{I_r}(x_{I_r}) + \frac{1}{2}\sum_{x \in (\X_1 \cap \X')} \hD^{(r)}(x) & \mbox{Defn. of $\X'$.}\\
\intertext{In order to minimize the above expression, we can now assume without loss of generality that the partition between $\X_0 \cap \X'$ and $\X_1 \cap \X'$ depends solely on $x_{I_r}$ as only the relative magnitude of $\hD^{(r)}_{I_r}(x_{I_r})/4$ and $1/2$ matters. Continuing,}
& \ge \frac{\beta^2}{4t} + \frac{\alpha-\beta}{2} \:\ge\: \frac{\alpha^2}{4t} \:\ge\: \frac{1}{16t}. & \mbox{Convexity of $x^2$.}
\end{align*}

Finally, we are ready to transfer the error probability back in exactly the same manner as we did in the proof of Lemma~\ref{lem:real-err-prob-et}. Using a similar argument, the existence of $\hM_j$ guarantees that
\[
\perr \ge \frac{1}{16t} \cdot \exp\left(-4r-\sum_{j=1}^r \gamma_j\right) = \exp\left(-cE-4-4r-\sum_{j=1}^r \gamma_j\right) > 2^{-E},
\]
since
\[
cE+4+4r+\sum_{j=1}^r \gamma_j \le cE + \frac{14cE}{100} + \frac{32cE}{100} + \frac{44cE}{100} \:<\: E.
\]
\end{proof}

\begin{proof}[Proof of Theorem~\ref{thm:lb-ee}]
Similarly to the proof of Theorem~\ref{thm:lb-et}, Theorem~\ref{thm:lb-ee} follows from Lemma~\ref{lem:real-err-prob-ee} and a direct application of Yao's minimax principle.
\end{proof}

\section{New Protocols for $\EqualityTesting$ and $\ExistsEqual$} \label{sect:ub}

In this section, we attempt to prove that our  $\Omega(Ek^{1/r})$ lower bound is \emph{tight} for $\EqualityTesting$.  We manage to attain this bound in several situations, but fail to achieve it for every value of $E,k,r$.

First of all, the $\Omega(Ek^{1/r})$ bound is only binding when it is at least $\Omega(k)$, which is necessary even when $E$ is constant~\cite{KalyanasundaramS92,Razborov92,DasguptaKS12}.
In Theorem~\ref{thm:ub-s1} we give a $\log^*(k/E)$-round protocol that reduces the effective dimension of the problem from $k$ to at most $E$ with $O(k)$ communication,
and basically lets us proceed under the assumption that 
$E\ge k$. (Note that if $E\ge k$ initially, $\log^*(k/E)=0$.)

In Theorem~\ref{thm:ub-et-s2-1} we give a simple protocol
for $\EqualityTesting$ with communication $O(rEk^{1/r})$ when $E\ge k$.  According to Theorem~\ref{thm:lb-et} this is optimal when $r=O(1)$.  All of our remaining protocols aim to eliminate or reduce this seemingly unnecessary factor of $r$.  In Theorem~\ref{thm:ub-ee-s2} we prove that $\ExistsEqual$ can be solved with $O(Ek^{1/r})$ communication, for any 
$r$ and $E\ge k$, and Theorem~\ref{thm:ub-et-s2-2} 
shows the same communication can be attained for $\EqualityTesting$, but with $O(r)$ rounds rather than $r$.
In particular, Theorems~\ref{thm:ub-s1}, \ref{thm:ub-ee-s2}, and \ref{thm:ub-et-s2-2} imply that $\EqualityTesting/\ExistsEqual$ can be solved with 
absolutely optimal communication $O(k+E)$ 
in $\log k$ rounds, which is also round-optimal 
according to Theorems~\ref{thm:lb-et} and \ref{thm:lb-ee}.
However, Theorems~\ref{thm:lb-et},~\ref{thm:ub-et-s2-1}, and~\ref{thm:ub-et-s2-2} leave the precise complexity of
$\EqualityTesting$ open when $E\ge k$ and $r$ is between
$\omega(1)$ and $o(\log k)$.

Theorem~\ref{thm:ub-et-s3} is our most sophisticated upper bound, in many ways.  It proves that $\EqualityTesting$ can be solved using $O(Ek^{1/r}\log r + Er\log r)$ communication when $E\ge k$.  When $r \ge \log k/\log\log k$ the first term is dominant, and the protocol comes within a $\log r \leq \log\log k$ factor of Theorem~\ref{thm:lb-et}'s lower bound.
Taken together, these theorems highlight a potential complexity separation between $\ExistsEqual$ and $\EqualityTesting$ and between $\SetDisjointness$ and $\SetIntersection$ in the low error probability regime. 

Theorems~\ref{thm:ub-et-1}--\ref{thm:ub-et-3} 
follow by combining the
dimension reduction of Theorem~\ref{thm:ub-s1} with Theorems~\ref{thm:ub-et-s2-1}--\ref{thm:ub-et-s3}.

\begin{theorem}\label{thm:ub-et-1}
There exists a $(\log^*(k/E)+r)$-round randomized protocol for
$\EqualityTesting$ on vectors of length $k$ that errs with probability $\perr = 2^{-E}$, using $O(k+rEk^{1/r})$ bits of communication.
\end{theorem}

\begin{theorem}\label{thm:ub-ee}
There exists a $(\log^*(k/E)+r)$-round randomized protocol for
{\ExistsEqual} on vectors of length $k$ that errs with probability $\perr = 2^{-E}$, using $O(k+Ek^{1/r})$ bits of communication.
\end{theorem}

\begin{theorem}\label{thm:ub-et-2}
There exists a $(\log^*(k/E)+O(r))$-round randomized protocol for
{\EqualityTesting} on vectors of length $k$ that errs with probability $\perr = 2^{-E}$, using $O(k+Ek^{1/r})$ bits of communication.
\end{theorem}

\begin{theorem}\label{thm:ub-et-3}
There exists a $(\log^*(k/E)+r)$-round randomized protocol for
{\EqualityTesting} on vectors of length $k$ that errs with probability $\perr = 2^{-E}$, using $O(k+Ek^{1/r}\log r + Er\log r)$ bits of communication.
\end{theorem}

\medskip

\begin{remark}\label{rem:logstar}
The $\log^*(k/E)$ terms in the round complexity of Theorems~\ref{thm:ub-et-1}--\ref{thm:ub-et-3} are not absolute.  They
can each be replaced with $\max\{0,\, \log^*(k/E) - \log^*(C)\}$, 
at the cost of increasing the communication by $O(Ck)$.
\end{remark}

\begin{remark}\label{rem:SetDSetI-reductions}
By applying Theorem~\ref{thm:Equivalence-Int-Eq} to Theorems~\ref{thm:ub-et-1}--\ref{thm:ub-et-3} we obtain 
$\SetDisjointness/\SetIntersection$ protocols with the same communication complexity, but with one more round of communication.  In the case of $\SetDisjointness$ (Theorem~\ref{thm:Equivalence-Int-Eq} + Theorem~\ref{thm:ub-ee}), 
it is straightforward to skip the reduction of Theorem~\ref{thm:Equivalence-Int-Eq} 
and solve the problem directly with $O(k+Ek^{1/r})$ communication in $(r + \log^*(k/E))$ rounds.
However, we do not see how to avoid Theorem~\ref{thm:Equivalence-Int-Eq}'s 
extra round of communication 
when solving $\SetIntersection$.
I.e., the $\SetIntersection$ protocols implied
by Theorems~\ref{thm:Equivalence-Int-Eq}, \ref{thm:ub-et-1}, and \ref{thm:ub-et-3}
use $(\log^*(k/E) + r \, \underline{+\, 1})$ rounds.
\end{remark}

\subsection{Overview and Preliminaries}\label{sect:ub-prelim}

We start by giving a generic protocol for $\EqualityTesting$. The protocol uses a simple subroutine for $\ExistsEqual/\EqualityTesting$ when $k = 1$.
Suppose Alice and Bob hold $x, y \in U = \{0,1\}^l$, respectively.
Alice picks a random $w \in \{0,1\}^l$ from the shared random source
and sends Bob $\check{x} = \left<x, w\right> \bmod 2$,
where $\left<\cdot,\cdot\right>$ is the inner product operator.
Bob computes $\check{y} = \left<y, w\right> \bmod 2$
and declares ``$x=y$'' iff $\check{x}=\check{y}$. 
Clearly, Bob never errs if $x=y$;
it is straightforward to show that the probability of error is exactly $1/2$ when $x \neq y$.
We call this protocol an \emph{inner product test} and $\check{x},\check{y}$ \emph{test bits}. 
A \emph{$b$-bit inner product test} on $x$ and $y$ refers to $b$ independent inner product tests on $x$ and $y$.    

At the beginning of phase $j$, $j \geq 1$, 
Alice and Bob agree on a subset $I_{j-1}$ of coordinates
on which all previous inner product tests have passed.
In other words, they have \emph{refuted} the 
potential equality $x_i\stackrel{?}{=}y_i$ for all $i \in [k]\backslash I_{j-1}$.
Each coordinate $i \in I_{j-1}$ represents either an actual equality ($x_i = y_i$),
or a \emph{false positive} ($x_i \neq y_i$).
At the beginning of the protocol, $I_0 = [k]$.
In phase $j$, we perform $l_j$ independent inner product tests on each coordinate in $I_{j-1}$
and let $I_j \subseteq I_{j-1}$ be the remaining coordinates that pass all their respective inner product tests.
Notice that each coordinate in $I_{j-1}$ corresponding to equality will always pass all the tests and enter $I_j$, while those corresponding to inequalities will only enter $I_j$ with probability $2^{-l_j}$.
At the end of the protocol, we declare all coordinates in $I_r$ \emph{equal} and all other coordinates \emph{not equal}.

This finishes the description of the generic $\EqualityTesting$ protocol. Theorems~\ref{thm:ub-et-1}--\ref{thm:ub-et-s3} 
all build on the framework of the generic protocol, instantiating its steps in different ways.

\subsubsection{A Protocol for Exchanging Test Bits}
 
For $\EqualityTesting$, it is possible that a constant fraction 
of the coordinates are actually equalities, which makes $|I_j| {} = \Theta(k)$ for every $j$. 
The naive implementation would explicitly exchange all $l_j|I_{j-1}|$ test bits and 
use $\Omega(kE)$ bits of communication in total. 
All the test bits corresponding to equalities are ``wasted'' in a sense. 

For our application, it is important that the communication volume that Alice and Bob use to exchange their test bits in phase $j$ be proportional to the number of false positives in $I_{j-1}$, instead of the size of $I_{j-1}$.
We will use a slightly improved version of a protocol of 
Feder et al.~\cite{FederKNN95} for exchanging the test bits.

Imagine packing the test bits into vectors 
$\hat{x},\hat{y} \in B^{|I_{j-1}|}$ where $B=\{0,1\}^{l_j}$.
Lemma~\ref{lem:hamming-distance} shows that Alice can 
transmit $\hat{x}$ to Bob, at a cost that depends on an \emph{a priori}
upper bound on the Hamming distance $\dist(\hat{x},\hat{y})$, i.e.,
the number of the coordinates in $I_{j-1}$ where they differ.

\begin{lemma}[Cf.~Feder et al.~\cite{FederKNN95}.]\label{lem:hamming-distance}
Suppose Alice and Bob hold length-$K$ vectors $x,y \in B^K$, where
$B=\{0,1\}^L$.  Alice can send one $O(dL + d\log(K/d))$-bit message to Bob,
who generates a string $x'\in B^K$ such that the following holds.
If the Hamming distance $\dist(x,y) \le d$ then $x=x'$; 
if $\dist(x,y) > d$ then there is no guarantee.
\end{lemma}

\begin{proof}
Define $G = (V,E)$ to be the graph on 
$V=B^K$ such that $\{u,v\}\in E$ iff
$\dist(u,v)\leq 2d$.  The maximum degree in $G$ is clearly
at most $\Delta =  \binom{K}{2d}\cdot 2^{2Ld}$ 
since there are $\binom{K}{2d}$ ways
to select the $2d$ indices and $2^{2Ld}$ ways to change
the coordinates at those indices so that there are \emph{at most} $2d$ different coordinates.  
Let $\phi : V\mapsto [\Delta]$
be a proper $\Delta$-coloring of $G$.
Alice sends $\phi(x)$ to Bob, which requires 
$\log\Delta = O(dL + d\log(K/d))$ bits.  Every string in
the radius-$d$ ball around $y$ (w.r.t.~$\dist$) is
colored differently since they are all at distance at most $2d$,
hence if $\dist(x,y) \le d$, Bob can reconstruct $x$ without error.
\end{proof}

\begin{corollary} \label{cor:hash-test}
Suppose at phase $j$, it is guaranteed that the number of false
positives in $I_{j-1}$ is at most $k_{j-1}$.
Then phase $j$ can be implemented with 
$O(k_{j-1}l_j + k_{j-1} \log(k/k_{j-1}))$ bits in $2$ rounds.
\end{corollary}

Finally, a naive implementation of the protocol requires $2r$ rounds
if the generic protocol has $r$ phases.
In fact, the protocol can be compressed into exactly $r$ rounds in the following way.
At the beginning, both parties agree that $I_0 = [k]$.
Alice generates her $l_1|I_0|$ test bits $\hat{x}^{(1)}$ for phase $1$ and communicates them to Bob;
Bob first generates his own test bits $\hat{y}^{(1)}$ for phase $1$ and determines $I_1$, then generates $l_2|I_1|$ test bits $\hat{y}^{(2)}$ for phase $2$ and transmits \emph{both} $\hat{y}^{(1)}$ and $\hat{y}^{(2)}$ to Alice.
Alice computes $I_1$, generates $\hat{x}^{(2)}$, computes $I_2$, generates $\hat{x}^{(3)}$,
and then sends $\hat{x}^{(2)}$ and $\hat{x}^{(3)}$ to Bob, and so on.
There is no asymptotic increase in the communication volume.

\subsubsection{Reducing the Number of False Positives}

Our protocols for $\EqualityTesting$ and $\ExistsEqual$ are 
divided into two parts.  The goal of the first part is to reduce
the number of false positives from at most $k$ to at most $E$;
if $E\ge k$, we can skip this part.
The details of this part are very similar to 
\Saglam{} and Tardos's $\SetDisjointness$ protocol~\cite{SaglamT13}.

\begin{theorem}\label{thm:ub-s1}
Let $(x,y)$ be an instance of $\ExistsEqual$ with $|x|{} ={} |y|{} = k$.
In $\log^*(k/E)$
rounds, we can 
reduce this to a new instance $(x',y')$ of $\ExistsEqual$ 
where $|x'|{} ={} |y'|{} \le E$, using $O(k)$ communication. 
The failure probability of this protocol is at most $2^{-(E+1)}$.

For $\EqualityTesting$, we can reduce the initial instance to a new 
instance $(x', y')$ such that the Hamming distance $\dist(x',y') \le E$, 
with the same round complexity, communication volume, and error probability.
\end{theorem}

\begin{proof}
We first give the protocol for $\ExistsEqual$, then apply the necessary changes
to make it work for $\EqualityTesting$.

The protocol for $\ExistsEqual$ uses our generic protocol, and imposes a strict upper bound $k_j$ on $|I_j|$. 
Whenever $|I_j|$ exceeds this upper bound, we halt the entire protocol and answer \emph{yes}
(there exists a coordinate where the input vectors are equal).
We start by setting the parameters $k_j$ and $l_j$ for any $j \in [1,\log^*(k/E)]$ as follows.
\begin{align*}
k_0 & = k,\\
k_j & = \max\left\{\frac{k}{2^{j-1}\exp^{(j)}(2)},E\right\}, \\
l_j & = 3 + \exp^{(j-1)}(2).
\end{align*}
Note that it is reasonable to assume $k_j > E$ before the last phase, since whenever we find $k_j \leq E$, we can simply terminate the protocol prematurely after phase $j$, and our goal would be achieved.

Now suppose the input vectors share no equal coordinates. 
We know that $|I_{j-1}| {}\le k_{j-1}$ at the beginning of phase $j$.
The probability of any particular coordinate in $I_{j-1}$ passing all tests in phase $j$ is exactly $p_j = \exp(-l_j)$. 
Thus, the expected size of $I_j$ is at most
\[
k_{j-1}p_j = \frac{k}{2^{j-2}\exp^{(j-1)}(2)} \cdot \frac{1}{2^3\exp^{(j)}(2)} \le \frac{k}{2^{j+2}\exp^{(j)}(2)} \le \frac{k_j}{8}.
\]

Recall the statement of the usual Chernoff bound.

\begin{fact}[See \cite{Dubhashi09}]
Let $X = \sum_{i=1}^n X_i$, where each $X_i$ is an i.i.d.~Bernoulli random variable.
Letting $\mu = \E[X]$, the following inequality holds for any $\delta > 0$.
\[
\Pr[X \ge (1+\delta)\mu] \le \left(\frac{e^{\delta}}{(1+\delta)^{1+\delta}}\right)^{\mu}.
\]
\end{fact}
In our case $X_i=1$ iff the $i$th coordinate in $I_{j-1}$ survives to $I_j$.
By linearity of expectation, $\mu \le k_j/8$.
Setting $\delta = k_j/\mu - 1 \ge 7$, we have
\[
\Pr[X \ge k_j]
= \Pr[X \ge (1+\delta)\mu] \\
 \le \left(\frac{e^{\delta}}{(1+\delta)^{1+\delta}}\right)^{\frac{k_j}{1+\delta}} 
< \left(\frac{e^7}{8^8}\right)^{k_j/8}
< 2^{-1.7k_j}.
\]
Hence, the probability that there are at least $k_j$ coordinates remaining after phase $j$ is at most $2^{-1.7k_j} \leq 2^{-1.7E}$, and the probability
this happens in any phase is at most $\sum_j 2^{-1.7k_j} \le 2^{-(E+1)}$.
Notice that when $x$ and $y$ share at least one equal coordinate, the error probability of this protocol is $0$ because if it fails to reduce
the number of coordinates to $E$ it (correctly) answers \emph{yes}.
The communication volume of the protocol is asymptotic to
\[
\sum_j l_j|I_{j-1}| 
{}\leq {} \sum_j l_jk_{j-1}
{}= {} \sum_j O(k/2^j) 
{}= {}O(k).
\]

For $\EqualityTesting$, we use the same $k_j$ as an upper bound on the number of false positives in $I_{j}$, instead of the size of $I_j$. Since the number of false positives is at most $k$ at the beginning, we can still use the same argument to show that with the same choice of $k_j$ and $l_j$, after $\log^*(k/E)$ phases, the number of false positives is at most $E$ with error probability $2^{-(E+1)}$.
By Lemma~\ref{lem:hamming-distance}, the number of bits we need to exchange in phase $j$ 
is $O(k_{j-1} l_j + k_{j-1} \log (k / k_{j-1}))$. 
Notice that $\log(k / k_{j - 1}) = j - 2 + \exp^{(j - 2)}(2) = O(\log l_j)$, 
so the total communication volume is still $O(k)$. 
\end{proof}

In all of our protocols, we first apply Theorem~\ref{thm:ub-s1} to reduce the number 
of coordinates (in the case of $\ExistsEqual$) 
or false positives (in the case of $\EqualityTesting$)
to be at most $E$. 
This requires no communication if $E\ge k$ to begin with.
Hence, with $\log^*(k/E)$ extra rounds and $O(k)$ communication, we will assume henceforth that all instances
of $\ExistsEqual$ have $k\leq E$ and all instances
of $\EqualityTesting$ have $\dist(x,y) \leq E$.

\subsection{A Simple {\EqualityTesting} Protocol}\label{sect:ub-et-1}

In light of Theorem~\ref{thm:ub-s1}, 
we can assume that the input vectors to $\EqualityTesting$ 
are guaranteed to differ in at most $k_0 = \min\{k,E\}$ coordinates.

\begin{theorem}\label{thm:ub-et-s2-1}
Fix any $k \ge 1$, $E \ge 1$, and $r \in [1,(\log k_0)/2]$, where $k_0 = \min\{k,E\}$. 
There exists a randomized protocol for {\EqualityTesting} length-$k$ vectors $x,y$
with Hamming distance $\dist(x,y)\leq k_0$ that uses $r$ rounds, 
$O(k + rEk_0^{1/r})$ bits of communication,
and errs with probability $\perr = 2^{-(E+1)}$.
\end{theorem}

\begin{proof}
We instantiate the generic protocol.  The parameter $l_j$ is the number
of test bits generated per coordinate of $I_{j-1}$ in phase $j$.
The parameter $k_j$ is an upper bound on the number of 
false positives surviving in $I_j$
(with high probability $1-2^{-\Theta(E)}$).
\begin{align*}
k_j & = k_0^{1-j/r}, \\
l_j & = 4Ek_0^{j/r-1}.
\end{align*}

Now fix a phase $j \in [1,r]$ and suppose at the beginning of phase $j$ that the number of false positives in $I_{j-1}$ is at most $k_{j-1}$.  By assumption this holds
for $j=1$.  The probability that at least $k_j$ false positives
survive phase $j$ is upper bounded by
\begin{align*}
\binom{k_{j-1}}{k_j} 2^{-k_jl_j}
& \le \left(\frac{ek_{j-1}}{k_j}\right)^{k_j} 2^{-k_jl_j} & \tag{$\binom{n}{k} \le \left(\frac{en}{k}\right)^k$}\\
& \le 2^{2k_j\log(k_{j-1}/k_j)-k_jl_j} & \tag{$e \le k_0^{1/r} = \frac{k_{j-1}}{k_j}$ due to $r \le \frac{\log k_0}{2}$}\\
& \le 2^{-2E}. &\tag{$\log\frac{k_{j-1}}{k_j} \le \frac{k_{j-1}}{k_j} = k_0^{1/r} \le \frac{l_j}{4}$}
\end{align*}
Thus, by a union bound, the number of false positives surviving phase $j$
is \emph{strictly less than} $k_j$, for all $j\in [1,r]$,
with probability at least $1-2^{-(E+1)}$.
In particular, there are no false positives at the end since $k_r = 1$.

Meanwhile, by Lemma~\ref{lem:hamming-distance}, the total communication volume is $O(k + rEk_0^{1/r})$ since
\begin{align*}
\sum_{j=1}^r k_{j-1}l_j
& = 4rEk_0^{1/r}, \\
\intertext{and}
\sum_{j=1}^r k_{j-1}\log\frac{k}{k_{j-1}}
& = k_0\sum_{j=0}^{r-1} \frac{1}{k_0^{j/r}}\left(\log\frac{k}{k_0} + \log k_0^{j/r}\right) \\
& \le 2k_0\log\frac{k}{k_0} + k_0\sum_{j=0}^{r-1}\frac{\log k_0^{j/r}}{k_0^{j/r}} & \tag{$k_0^{1/r} \ge 2^2$ due to $r \le \frac{\log k_0}{2}$}\\
& = O(k). & \tag{$k_0^{1/r} \ge 2^2$ and $k_0 \le k$}
\end{align*}
\end{proof}

\begin{proof}[Proof of Theorem~\ref{thm:ub-et-1}]
Applying Theorem~\ref{thm:ub-s1} and Theorem~\ref{thm:ub-et-s2-1} in sequence, 
we obtain a $(\log^*(k/E)+r)$-round randomized protocol for {\EqualityTesting} 
on vectors of length $k$ that errs with probability 
$\perr = 2^{-E}$ and uses $O(k+rE\min\{k,E\}^{1/r})$ bits of communication.
When $E\ge k$ the protocol is obtained directly from 
Theorem~\ref{thm:ub-et-s2-1} and uses $O(rEk^{1/r})$ communication.
When $E < k$ the communication implied by Theorems~\ref{thm:ub-s1}
and~\ref{thm:ub-et-s2-1} is $O(k + rE^{1+1/r}) = O(k + rEk^{1/r})$.\footnote{It appears as if $rE^{1+1/r}$ is an improvement over $rEk^{1/r}$ when $E<k$,
but this is basically an illusion.
In light of Remark~\ref{rem:logstar}, 
we can always dedicate $\log^*(k/E)-2$ rounds to the first part
and $r+2$ rounds to the second part while increasing the communication by $O(k)$.
When $E\ge k^{1-1/r}$, $rE^{1+1/r} = \Omega((r+2)Ek^{1/(r+2)})$,
meaning there is no clear benefit to use the $rE^{1+1/r}$ expression.}
\end{proof}

\subsection{An Optimal {\ExistsEqual} Protocol}\label{sect:ub-ee}

\subsubsection{Overview of the Protocol}\label{sect:desc-ub-ee}

In this section, we show that we can obtain a $(\log^{*}(k/E) + r)$-round, 
$O(k + Ek^{1/r})$-bit protocol for $\ExistsEqual$.
This matches the lower bound of Theorem~\ref{thm:lb-ee}, asymptotically, when $E\ge k$. 
Theorem~\ref{thm:ub-s1} covers dimension reduction in $\log^*(k/E)$ rounds,
so we assume without loss of generality that $E\ge k$ and we have exactly $r$ rounds.

Suppose the inputs $x$ and $y$ share no equal coordinates.
Imagine writing down all the possible results of the inner product tests in a matrix 
$A$ of dimension $(E + \log k) \times k$, 
where $A_{q,i}$ is ``$=$'' if $x_i, y_i$ pass the 
$q$th inner product test, and ``$\neq$'' otherwise. 
By a union bound, with probability $1 - 2^{-E}$, each column contains at least one ``$\neq$''. 
Now consider the area above the first ``$\neq$'' in each column.  
The probability that this area is at least $E'$ is, by a union bound,
at most
\begin{equation}\label{eqn:error-budget}
\binom{E' + k-1}{k-1}2^{-E'} < \exp(k\log(e(E'+k)/k) - E').
\end{equation}
For $E' = E + O(k\log(E/k)) = O(E)$, this probability is $\ll 2^{-E}$.
In our analysis it suffices to consider a situation where an \emph{adversary}
can decide the contents of $A$, subject to the constraint that its \emph{error budget}
(the area above the curve defined by the first ``$\neq$'' in each column) never exceeds $E'=O(E)$.
The notion of an error budget is also essential for analyzing the $\EqualityTesting$ protocols of Section~\ref{sect:ub-et-2} Section~\ref{sect:ub-et-3}.

In the $j$th phase, $j \geq 1$, our protocol exposes the fragment of $A$ consisting of the 
next $l_j$ rows of columns in $I_{j-1}$.  The set $I_j$ consists of those columns without any ``$\neq$'' exposed so far.
The \emph{communication budget} for phase $j$ is equal to $l_j|I_{j-1}|$.
In the worst case, the first exposed value in each column of $I_{j-1}\setminus I_j$ is ``$\neq$'',
so the adversary spends at least $l_j|I_j|$ of its \emph{error budget} in phase $j$.

If we witness at least one ``$\neq$'' in every column, we can correctly declare there does not exist an 
equal coordinate and answer \emph{no}. Otherwise, if the adversary has not exceeded his error budget 
but there is some column without any ``$\neq$'', we answer \emph{yes}. 
If the adversary ever exhausts his error budget, we terminate the protocol and answer \emph{yes}.
Recall that the notion of an error budget tacitly assumed that $x$ and $y$ differ in every coordinate.  It is important to note that if they do \emph{not} 
differ in every coordinate, the protocol answers correctly 
with probability 1, regardless of whether the 
protocol halts prematurely or not.  Thus, there is nothing to prove in this case
and it is fine to measure the error budget expended \emph{as if} Alice's and Bob's 
inputs differ in all coordinates.
The probability that the error budget is exhausted when $x$ and $y$ differ in all coordinates (causing the algorithm 
to incorrectly answer \emph{yes})
is $\ll 2^{-E}$, according to Eqn.~(\ref{eqn:error-budget}).

\subsubsection{Analysis}

In this section we give a formal proof to the following Theorem:

\begin{theorem}\label{thm:ub-ee-s2}
Fix any $k \ge 1$, $E \ge k$, and $r \in [1,(\log k)/2]$. There exists an $r$-round randomized protocol for $\ExistsEqual$ on vectors of length $k$ that errs with probability $\perr=2^{-(E+1)}$, using $O(Ek^{1/r})$ bits of communication.
\end{theorem}

\begin{proof}
The number of tests per coordinate in phase $j$ is $l_j$:
\begin{align*}
l_j & = 2Ek^{j/r-1}.
\end{align*}
Define $E_j = \sum_{j' = 1}^{j} l_{j'}|I_{j'}|$ to be the portion of the error budget spent in phases 1 through $j$.
We can express the asymptotic 
communication cost of the protocol in terms of the error budget as follows.
\begin{align*}
\sum_{j=1}^r l_j |I_{j-1}|
& \le l_1|I_0|{} + k^{1/r}\sum_{j=2}^{r} l_{j-1} |I_{j-1}| & \mbox{$l_j = k^{1/r} l_{j-1}$.}\\
& \le 2Ek^{1/r} + E_{r-1}k^{1/r} & \mbox{Defn. of $E_{r-1}$.}\\
\intertext{Recall that the protocol terminates immediately after phase $j$ if $E_j \ge E'$, which indicates $E_{r-1} < E'$.  Hence, the total cost is bounded by}
& \le (2E+E')k^{1/r} \:=\: O(Ek^{1/r}).
\end{align*}

The protocol can only err if $x$ and $y$ differ in every coordinate.
In this case, there are two possible sources of error.
The first possibility is that the protocol answers \emph{yes} because $|I_r| {}\ge 1$.
By a union bound, this happens with probability at most
\[
k 2^{-\sum_{j=1}^r l_j} \leq k2^{-l_r} = k2^{-2E}.
\]

The second possibility is that the protocol terminates 
prematurely and answers \emph{yes}
if $E_j \ge E'$ for some $j \in [1,r]$.
The probability of this event occuring is also
$\ll 2^{-E}$; see Eqn.~(\ref{eqn:error-budget}).
This concludes the proof.
\end{proof}

\begin{proof}[Proof of Theorem~\ref{thm:ub-ee}]
Theorem~\ref{thm:ub-ee} follows directly by combining Theorem~\ref{thm:ub-s1} and Theorem~\ref{thm:ub-ee-s2}.
\end{proof}

\subsection{A $\log k$-Round Communication Optimal {\EqualityTesting} Protocol}\label{sect:ub-et-2}

Suppose we want a communication optimal $\EqualityTesting$ 
protocol using $O(k+E)$ bits.  When $E\ge k$ we 
need $r=\Omega(\log k)$ rounds, by Theorem~\ref{thm:lb-et}.
In this section, we give a protocol for $\EqualityTesting$ that uses $O(r)$ rounds (rather than $r$) 
and $O(Ek^{1/r})$ bits of communication, assuming $E\ge k$.
Observe that when $r=\Theta(\log k)$, 
there is no (asymptotic) difference between $r$ rounds
and $O(r)$ rounds as this only influences the leading
constant in the communication volume.

\subsubsection{Overview of the Protocol}\label{sect:desc-ub-et-2}

The protocol uses the concept of an \emph{error budget} 
introduced in Section~\ref{sect:ub-ee}. 
To shave the factor $r$ off the communication volume, 
we cannot afford to use $Ek^{j/r - 1}$ test bits 
for each coordinate that participates in phase $j$. 
Consequently, we \emph{cannot guarantee} with high probability (say $1 - 2^{-\Theta(E)}$) 
that the number of false positives in $I_j$ is less than $k^{1 - j/r}$.

Our protocol needs to be able to respond to the rare event that the number of false positives in $I_j$ is larger than $k_j$.
Notice that this type of error cannot be detected in the 
first $j$ phases, and is not easily detectable in the following phases.  The danger in the number of false positives in $I_j$ exceeding $k_j$
is that when the test bits for phase $j+1$ 
are exchanged using Lemma~\ref{lem:hamming-distance},
the protocol may silently fail, 
with all test bits potentially corrupted.

To address these challenges, 
Alice and Bob each keep a \emph{history} of all the test bits they have generated so far. They also keep a history of the test bits they have received from the other party, 
\emph{which may have been corrupted}. 
Define $T_{A}$ and $T_{B}$ to be the true history of the test bits generated by Alice and Bob, 
respectively.  Define $T_{B}^{(A)}$ to be what Alice believes Bob's history to be,
and define $T_{A}^{(B)}$ analogously.  Observe that if every invocation of Lemma~\ref{lem:hamming-distance} succeeds, then $T_A = T_A^{(B)}$ and $T_B = T_B^{(A)}$.

To detect inconsistencies, after Alice and Bob generate and exchange their test bits for phase $j$,
they accumulate their views of the history into 
strings $T^{(A)} = T_A \circ T_B^{(A)}$
and $T^{(B)} = T_A^{(B)} \circ T_B$, respectively,
where $\circ$ is the concatenation operator,
and verify that $T^{(A)}=T^{(B)}$ with a certain number of inner product tests.  
This is called a \emph{history check}.
If the history check passes, 
they can proceed to phase $j+1$.
If the history check fails then the results of phase $j$
are junk, and we can infer that one of two types of 
low probability events occurred in phase $j-1$.
The first possibility is that the test bits at phase
$j-1$
were exchanged successfully (and consequently,
the history check succeeded), 
but $I_{j-1}$ contains more than $k_{j-1}$ false positives.  The second possibility is that
Alice's and Bob's histories were already inconsistent
at phase $j-1$, but the phase-$(j-1)$ 
history check failed to detect this.
Notice that Alice and Bob cannot detect
which of these types of errors occurred. 
In either case, we must undo the effects of 
phases $j$ and $j - 1$ and restart the protocol
at the beginning of phase $j-1$.
It may be that the history check then fails at the
re-execution of phase $j-1$, in which case we would continue to rewind to the beginning of phase 
$j-2$, and so on.
Being able to rewind multiple phases is important 
because we do not know which phase suffered the \emph{first} error. 

Both parties maintain an empirical 
\emph{error meter} $E''$ that measures the sum of logarithms of probabilities of low probability (error) events that have been detected.  If the error meter ever exceeds the error budget $E' = \Theta(E)$ we terminate the protocol, which we show occurs with probability $\ll 2^{-E}$. 
Thus, the process above (proceeding iteratively with phases, undoing and redoing them when errors are detected) must end by either successfully completing phase $r$ or exceeding the error budget.

If Alice and Bob successfully finish phase $r$, 
we are still not done. 
This is because an error can happen in the later phases 
but we do not have sufficiently high ($1-2^{-E}$) confidence that they all succeeded.
To build this confidence, 
Alice and Bob do inner product tests on the whole
history, gradually increasing their number until $\Theta(E)$ tests have been done.  If one of these history checks fails, we increase the error meter $E''$ appropriately and rewind the protocol to a suitable phase $j$ in the
first stage of the protocol.

Let us make every step of this protocol more quantitatively precise.
\begin{itemize}
    \item The protocol has two stages, the \emph{Refutation Stage} 
    (in which potential equalities are refuted)
    and the \emph{Verification Stage}, each consisting of a series of \emph{phases}.  Although the Refutation Stage logically precedes
    the Verification Stage, because phases can be undone, an execution
    of the protocol may oscillate between Refutation and Verification
    multiple times.
    \item The Refutation Stage is similar to the protocol in Section~\ref{sect:ub-et-1} except Alice and Bob will verify whether
    the messages conveyed by Lemma~\ref{lem:hamming-distance} are successfully
    received with further inner product tests.
    The \emph{budget} of phase $j$ is
    \[
    B_j = \frac{Ek_0^{1/r}}{\min\{j^2, r\}}.
    \]
    Observe that the sum of budgets, $\sum_{j'=1}^r B_{j'}$, is $O(Ek_0^{1/r})$.
    Thus, in phase $j$, we perform $l_j = B_j/k_{j-1}$ independent inner product tests
    on each coordinate in $I_{j-1}$, and exchange test bits with a Hamming distance
    of $k_{j-1}$.  I.e., we are working under the assumption (perhaps false)
    that there are $k_{j-1}$ false positives still in $I_{j-1}$.
    As usual, $I_0$ is initially $[k]$ 
    and
    \begin{align*}
        k_0 & \leq E\\
        k_j & \leq k_0^{1-j/r}
    \end{align*}
    All histories $T_{A}, T_B, T_{B}^{(A)}, T_{A}^{(B)}$ are initially empty, 
    and the error meter $E''$ is initially zero.
    \item Phase $j$ has two steps, the \emph{test step} and the \emph{history check step}. In the test step, Alice and Bob conduct inner product tests
    as in Section~\ref{sect:ub-et-1}, i.e., they generate $l_j$ test bits for 
    each coordinate in $I_{j - 1}$ and exchange them using Lemma~\ref{lem:hamming-distance}, assuming their Hamming distance is at most $k_{j-1}$.  Alice appends the test bits she generates onto the history $T_{A}$, and appends the test bits she receives from Bob 
    onto $T_{B}^{(A)}$. Bob does likewise.  In the history check step,
    they use $B_j$ independent inner product tests to check whether $T^{(A)}=T^{(B)}$,
    where $T^{(A)} = T_{A} \circ T_B^{(A)}$, and
    $T^{(B)} = T_A^{(B)} \circ T_B$. 
    The history check \emph{fails} if they detect inequality 
    and \emph{passes} otherwise.  Since $B_j$ is, in general, less than $E$,
    we are still skeptical of history checks that pass.
    \item If the history check for phase $j$ passes, Alice and Bob proceed to  phase $j+1$, or proceed to the Verification Stage if $j=r$. 
    Otherwise, an error has been detected: either the number of false
    positives in $I_{j-1}$ is at least $k_{j-1}$, or the history
    check at phase $j-1$ \emph{mistakenly} passed.  The latter occurs with probability $\exp(-B_{j-1})$ and we show the former occurs with probability 
    $\exp(-3k_0^{-1/r}B_{j-1}/4)$.  Not knowing which occurred, we
    increment the error meter $E''$ by $k_0^{-1/r}B_{j-1}/2$ due to a union bound.  If $E''$ exceeds 
    the \emph{error budget} $E' = cE$ then we halt, where $c\geq 2$ is 
    a suitable constant.  
    Otherwise we retract the effects of phases $j$ and $j-1$
    and continue the protocol at the beginning of phase $j-1$, 
    with ``fresh'' random bits so as not to recreate previous errors.
    \item Observe that after phase $r$ of the Refutation Stage,
    each coordinate in $I_r$ has only passed about $B_r/k_{r-1} = E/r$ inner product tests, which is not high enough.  
    Before the Verification Stage begins, Alice and Bob
    each generate $E'$ test bits for each coordinate in $I_r$
    and append them to $T^{(A)}$ and $T^{(B)}$.  
    (This can be viewed as a degenerate instantiation of Lemma~\ref{lem:hamming-distance} with $d=0$,
    which requires no communication.)
    If there are no false positives in $I_r$, these test bits must be identical.
    \item In the Verification Stage the phases are indexed in reverse order: $r,r-1,\ldots,1$.  In each successive phase $j$, Alice and Bob test
    the equality $T^{(A)} = T^{(B)}$ with $B_j$ independent inner product tests.  This process
    stops if it passes a total of $E'$ tests, in which case they report
    that $x$ and $y$ are \emph{equal} on $I_r$ and \emph{not equal} on $[k]\backslash I_r$,
    or some Verification phase $j$ detects that 
    $T^{(A)} \neq T^{(B)}$.  In this case,
    we know Verification phases $r,r-1,\ldots,j+1$ passed \emph{in error}, 
    and that there must
    also have been an error in Refutation phase $r$.  Therefore, 
    Alice and Bob increment $E''$ by $k_0^{-1/r}B_r/2 + \sum_{j' = j+1}^r B_{j'}$
    and halt if $E'' \ge E'$.  If not, they rewind the execution of the protocol
    to phase $j$ of the Refutation Stage and continue.
\end{itemize}

Algorithm~\ref{alg:et-s2-2-p1} recapitulates this description in the form
of pseudocode, from the perspective of Alice.  
Here $T_A[j,i]$ refers to the sequence of Alice's test bits in $T_A$ 
for the $i$th coordinate produced in the most recent execution of phase $j$, and $T_{A}[j_1 \cdots j_2, \cdot]$ refers to the test bits generated from phase $j_1$ to phase $j_2$. Phase $r+1$ refers to the $E'\times {} |I_r|$ test bits generated 
between the Refutation and Verification stages.
$T^{(A)}[j,i]$ refers to the concatenation of $T_{A}[j, i]$ and 
$T_{B}^{(A)}[j, i]$. 

\begin{algorithm}[!b]
\caption{An {\EqualityTesting} protocol for Theorem~\ref{thm:ub-et-s2-2} (from the perspective of Alice).}
\label{alg:et-s2-2-p1}
\begin{algorithmic}[1]
\Procedure{EqualityTesting}{}\Comment{main procedure}
    \State $I_0 \gets [k]$  
    \State $k_0 \gets \min\{k,E\}$ \Comment{initial bound on Hamming distance}
    \State $E' \gets cE$            \Comment{error budget}
    \State $E'' \gets 0$            \Comment{error meter}
    \For{$j \gets 1,\ldots,r$}
        \State $\displaystyle B_j \gets \frac{Ek_0^{1/r}}{\min\{j^2,r\}}$     \Comment{phase $j$ communication budget}
        \State $k_j \gets k_0^{1-j/r}$                          \Comment{\emph{ideal} upper bound on Hamming distance}
        \State $l_j \gets B_j/k_{j-1}$                          \Comment{tests per coordinate}
    \EndFor
    \State \Call{Refutation}{$1$}
    \State \Call{Verification}{$r$}
    \State Output \emph{equal} on coordinates $I_r$ and \emph{not equal} on $[k]\backslash I_r$
\EndProcedure
\algstore{ub-et-2}
\end{algorithmic}
\end{algorithm}

\addtocounter{algorithm}{-1}

\begin{algorithm}[!p]
\caption{An {\EqualityTesting} protocol for Theorem~\ref{thm:ub-et-s2-2} (from the perspective of Alice).(cont.)}
\label{alg:et-s2-2-p2}
\begin{algorithmic}[1]
\algrestore{ub-et-2}
\Procedure{InnerProductTest}{$w$,$b$}
    \State perform $b$ independent inner product tests on $w$
    and return the test bits
\EndProcedure
\Statex
\Procedure{Refutation}{$j$}\Comment{phase $j$ of the Refutation Stage}
    \State $T_A[j,\cdot] \gets \perp$\Comment{Clear test bits for phase $j$}
    \ForAll{$i \in I_{j-1}$}
        \State $T_A[j,i] \gets \Call{InnerProductTest}{x_i,l_j}$
    \EndFor
    \State send $T_A[j,\cdot]$ to Bob and receive $T_B^{(A)}[j,\cdot]$ from Bob via Lemma~\ref{lem:hamming-distance}
    \State $T^{(A)}[j,\cdot] \gets T_A[j,\cdot] \circ T_B^{(A)}[j,\cdot]$
    \State $\hat{T}^{(A)} \gets \Call{InnerProductTest}{T^{(A)}[1 \cdots j,\cdot],B_j}$ \label{line:ref-hash}
    \State send $\hat{T}^{(A)}$ to Bob and receive $\hat{T}^{(B)}$ from Bob directly
    \If{$\hat{T}^{(A)} = \hat{T}^{(B)}$}\Comment{passed history check}\label{line:ref-if}
        \State $I_j \gets \{i \in I_{j-1} \mid T_A[j,i] = T_B^{(A)}[j,i]\}$\Comment{all coords.~not yet refuted}
        \If{$j < r$}
            \State \Call{Refutation}{$j+1$}
        \Else
            \State $T^{(A)}[r+1,\cdot] \gets \perp$
            \ForAll{$i \in I_r$}
                \State $T^{(A)}[r+1,i] \gets \Call{InnerProductTest}{x_i,E'}$
            \EndFor
        \EndIf
    \Else\label{line:ref-else}
        \State $E'' \gets E'' + k_0^{-1/r}B_{j-1}/2$,
        and terminate if $E'' \ge E'$ \Comment{update error meter}\label{line:ref-err-budget}
        \State \Call{Refutation}{$j-1$}
    \EndIf
\EndProcedure
\Statex
\Procedure{Verification}{$j$}\Comment{phase $j$ of the Verification Stage}
    \State $\hat{T}^{(A)} \gets \Call{InnerProductTest}{T^{(A)}[\cdot,\cdot],B_j}$
    \State send $\hat{T}^{(A)}$ to Bob and receive $\hat{T}^{(B)}$ from Bob directly
    \If{$\hat{T}^{(A)} = \hat{T}^{(B)}$}\label{line:ver-if}
        \If{$\sum_{j'=j}^r B_{j'} < E'$}\label{line:ver-budget}\Comment{insufficiently confident to halt}
            \State \Call{Verification}{$j-1$}
        \EndIf
    \Else\label{line:ver-else}\Comment{error detected}
        \State $E'' \gets E'' + k_0^{-1/r}B_r/2 + \sum_{j'=j+1}^r B_{j'}$,
        and terminate if $E'' \ge E'$ \Comment{update error meter}\label{line:ver-err-budget}
        \State \Call{Refutation}{$j$}\Comment{rewind protocol to phase $j$}
        \State \Call{Verification}{$r$}
    \EndIf
\EndProcedure
\end{algorithmic}
\end{algorithm}

\subsubsection{Analysis}

To prove Theorem~\ref{thm:ub-et-2}, it suffices to prove the following Theorem~\ref{thm:ub-et-s2-2}.

\begin{theorem}\label{thm:ub-et-s2-2}
Fix any $k \ge 1$, $E \ge 1$, and $r \in [1,(\log k_0)/6]$, where $k_0 = \min\{k,E\}$. There exists 
a randomized protocol for {\EqualityTesting} length-$k$ vectors $x,y$ with Hamming distance $\dist(x,y) \leq k_0$
that uses $O(r)$ rounds, $O(k + Ek_0^{1/r})$ bits of communication, and errs with probability $\perr=2^{-(E+1)}$.
\end{theorem}

The protocol of Lemma~\ref{lem:hamming-distance} \emph{fails} if Bob does not generate the correct $x' = x$, 
which indicates that the precondition is not met, i.e., $\dist(x,y) > d$.
Refutation phase $j$ \emph{fails} if the condition in line~\ref{line:ref-if} is not satisfied and the else branch 
at line~\ref{line:ref-else} is executed in order to resume the protocol from phase $j-1$.
Similarly, we say Verification phase $j$ \emph{fails} if the condition in line~\ref{line:ver-if} is not satisfied, 
which also indicates the else branch at line~\ref{line:ver-else} is executed and the protocol is resumed from Refutation 
phase $j$.

We begin the proof by showing that the extra communication caused by redoing some of the Refutation/Verification phases is properly covered by the total error budget.
The following two lemmas actually prove that the error budget spent so far is correctly lower bounded in line~\ref{line:ref-err-budget} and line~\ref{line:ver-err-budget}, and then Lemma~\ref{lem:overhead} upper bounds the total number of extra phases by $O(r)$ and the overall extra communication by $O(k+Ek_0^{1/r})$.

\begin{lemma}\label{lem:test-err-budget}
Fix any $j \in [2,r]$. If phase $j$ of the Refutation Stage fails, then the outcome of the most recent execution of 
phase $j-1$ happened with probability at most $\exp(-k_0^{-1/r}B_{j-1}/2)$.
\end{lemma}

\begin{proof}
Recall that there are two types of errors at phase $j-1$.
If the $(j-1)$th history check erroneously passed, this occurred with probability $\exp(-B_{j-1})$.
The probability that more than $k_{j-1}$ false positives survive in $I_{j-1}$
is less than 
\begin{align*}
\binom{k_0}{k_{j-1}}2^{-k_{j-1}l_{j-1}}
& \le \left(\frac{ek_0}{k_{j-1}}\right)^{k_{j-1}}2^{-k_{j-1}l_{j-1}} & \tag{$\binom{n}{k} \le \left(\frac{en}{k}\right)^k$.}\\
& \le 2^{2k_{j-1}\log(k_0/k_{j-1})-k_{j-1}l_{j-1}} & \tag{$e \le k_0^{1/r} \le \frac{k_0}{k_{j-1}}$ due to $r \le \frac{\log k_0}{6}$}\\
& \le 2^{-3l_{j-1}k_{j-1}/4},
\end{align*}
where the last step follows from the inequality
\begin{align*}
2\log\frac{k_0}{k_{j-1}}
& = 2(j-1)\log k_0^{1/r} \\
& \le \frac{8^{j-1}\log k_0^{1/r}}{4(j-1)^2} & \tag{Because $8x^3 \le 8^x$ for $x \in \mathbb{N}$}\\
& \le \frac{k_0^{(j-1)/r}}{4(j-1)^2} & \tag{$\log k_0^{1/r} \le \frac{k_0^{1/r}}{8}$ due to $k_0^{1/r} \ge 2^6$}\\
& \le \frac{B_{j-1}}{4k_{j-2}} & \tag{Defn. of $B_{j-1}$}\\
& = \frac{l_{j-1}}{4}. & \tag{Defn. of $l_{j-1}$}
\end{align*}

Combining the above two cases, by a union bound, the outcome of the most recent execution of phase $j-1$ of the 
Refutation Stage happens with probability at most $\exp(-B_{j-1}) + \exp(-3l_{j-1}k_{j-1}/4) 
= \exp(-B_{j-1}) + \exp(-3k_0^{-1/r}B_{j-1}/4) \le \exp(-k_0^{-1/r}B_{j-1}/2)$, as claimed.
\end{proof}

\begin{lemma}\label{lem:check-err-budget}
Fix any $j \in [1,r]$. If phase $j$ of the Verification Stage fails, then the outcomes 
of the most recent execution of phases $r,r-1,\ldots,j+1$ of the Verification Stage 
and phase $r$ of the Refutation Stage happened with overall probability at most 
$\exp(-k_0^{-1/r}B_r/2 - \sum_{j'=j+1}^r B_{j'})$.
\end{lemma}

\begin{proof}
Notice that the failure of Verification phase $j$ means all previous Verification phases $r,r-1,\ldots,j+1$ 
failed to detect an inconsistency in the history,
which occurs with probability $\exp(-\sum_{j'=j+1}^r B_{j'})$.
Meanwhile, the inconsistency is caused by an error of some type in Refutation phase $r$,
which, according to Lemma~\ref{lem:test-err-budget}, occurs with probability at most $\exp(-k_0^{-1/r}B_r/2)$.
Therefore, the outcomes of the most recent execution of Verification phases $r,r-1,\ldots,j+1$ 
and Refutation phase $r$ 
happened with overall probability at most $\exp(-k_0^{-1/r}B_r/2 - \sum_{j'=j+1}^r B_{j'})$.
\end{proof}

\begin{lemma}\label{lem:overhead}
Algorithm~\ref{alg:et-s2-2-p1}
executes $O(r)$ extra Refutation/Verification phases
and uses $O(k + Ek_0^{1/r})$ extra bits of communication.
\end{lemma}

\begin{proof}
We first consider the total number of \emph{extra} phases.
Each failure of Refutation phase $j$ uses at least $k_0^{-1/r}B_{j-1}/2 \ge E/(2r)$ of the error budget 
and causes the re-execution of two phases, namely $j-1$ and $j$.
Similarly, each failure of Verification phase $j$ uses 
$k_0^{-1/r}B_r/2 + \sum_{j'=j+1}^r B_{j'} \ge (r-j+1)E/(2r)$ of the error budget and causes the re-execution of 
$2(r-j+1)$ phases.   Thus, the total number of extra phases is at most $4cr = O(r)$, where the error budget $E' = cE$.

Turning to the overall \emph{extra} communication, notice that phase $j$ of the Refutation Stage has communication volume 
$O(B_j + k_{j-1}\log(k/k_{j-1}))$ and phase $j$ of the Verification Stage has communication volume $O(B_j)$.
For any $j \in [2,r]$, also notice that $B_{j-1}/B_j \le j^2/(j-1)^2 \le 4 \le k_0^{1/r}$.
Thus, the communication caused by each failure is at most $O(k_0^{1/r})$ times the error budget 
spent by that failure, if we temporarily ignore the $k_{j-1}\log(k/k_{j-1})$ term.

In order to upper bound the communication contributed by the $k_{j-1}\log(k/k_{j-1})$ term, 
observe that Refutation phase $j$ can only be repeated $O(j^2)$ times before the error budget is exhausted.
Thus, the overall extra communication is upper bounded by $O(k + Ek_0^{1/r})$ since
\begin{align*}
\lefteqn{O(k_0^{1/r}) \cdot E' + \sum_{j=1}^r O(j^2)\cdot k_{j-1}\log\frac{k}{k_{j-1}}} \\
& = O(k_0^{1/r}) \cdot E' + k_0\sum_{j=1}^r \frac{O(j^2)}{k_0^{(j-1)/r}} \left(\log\frac{k}{k_0} + \log k_0^{(j-1)/r}\right) \\
& = O(k_0^{1/r}) \cdot E' + k_0 \log\frac{k}{k_0} \sum_{j=1}^r \frac{O(j^2)}{k_0^{(j-1)/r}} + k_0 \sum_{j=1}^r \frac{O(j^2)\cdot \log k_0^{(j-1)/r}}{k_0^{(j-1)/r}} \\
& = O(k + Ek_0^{1/r}). & \tag{$k_0^{1/r} \ge 2^6$ and $k_0 \le k$}
\end{align*}
\end{proof}

Now we are ready to prove Theorem~\ref{thm:ub-et-s2-2}.

\begin{proof}[Proof of Theorem~\ref{thm:ub-et-s2-2}]
If there are no errors, 
Algorithm~\ref{alg:et-s2-2-p1} has at most $2r$ phases and uses 
$O(\sum_{j=1}^r (B_j + k_{j-1}\log(k/k_{j-1}))) = O(k + Ek_0^{1/r})$ 
communication, where each phase can be implemented in $O(1)$ rounds.
Together with Lemma~\ref{lem:overhead}, we have shown that it is an $O(r)$-round 
randomized {\EqualityTesting} protocol using $O(k + Ek_0^{1/r})$ bits of communication.
Thus, it suffices to calculate the error probability of the protocol.

Consider a possible execution of the protocol, i.e., 
the sequence of the Refutation/Verification phases that are performed.
It can be represented by a \emph{unique} $0$-$1$ string of length at most 
$4cr+2r$ (by the proof of Lemma~\ref{lem:overhead}) such that each ``$1$'' corresponds to a failed phase.
In particular, each execution of the protocol that terminates prematurely because $E'' \ge E'$
is represented as a $0$-$1$ string, which occurs with probability at most $2^{-E'}$, by Lemmas~\ref{lem:test-err-budget} and \ref{lem:check-err-budget}.
Hence the overall probability of terminating prematurely is $2^{4cr+2r}\cdot 2^{-E'}$.

An error can also be caused by at least one false positive surviving all $E'$ independent inner product tests
generated after Refutation phase $r$.  The probability of this happening is at most $k_02^{-E'}$.
The last possible source of error is that all Verification phases fail to detect the 
inequality $T^{(A)}\neq T^{(B)}$.  According to line~\ref{line:ver-budget}, the probability of this happening 
is at most $2^{-E'}$.  Hence, the overall probability of error is upper bounded by
\[
2^{4cr+2r}\cdot 2^{-E'} + k_02^{-E'} + 2^{-E'} = \poly(k_0)2^{-E'},
\]
which is at most $2^{-E}$ for, say, $E' = 2E$. This concludes the proof.
\end{proof}

\begin{proof}[Proof of Theorem~\ref{thm:ub-et-2}]
Theorem~\ref{thm:ub-et-2} subsequently follows by applying 
Theorem~\ref{thm:ub-s1} and Theorem~\ref{thm:ub-et-s2-2} in sequence.
\end{proof}

\subsection{A More Efficient $\EqualityTesting$ Protocol} \label{sect:ub-et-3}

Theorem~\ref{thm:ub-ee-s2} demonstrates that the $\Omega(Ek^{1/r})$ lower bound can be
attained for $\ExistsEqual$.  Let us highlight a key property of the protocol that
arises naturally in $\ExistsEqual$ but is difficult to efficiently recreate in $\EqualityTesting$.
In the first round Alice generates about $l_1 = E/k^{1-1/r}$ test bits per coordinate
and sends them to Bob.  Since there is no possibility of 
reporting \emph{yes} ($\exists i.x_i=y_i$) in error, Bob can operate
under the assumption that $\forall i.x_i\neq y_i$.  Therefore, if he finds that 
$|I_1|=\beta_1 k^{1-1/r}$, he can infer that
the adversary has expended a $\beta_1$ fraction of his error budget
and \emph{adaptitvely} choose the length of his message to be
$\beta_1 Ek^{1/r}$, i.e., we are effectively charging $k^{1/r}$ bits of communication to each of the $\beta_1 E$ units of error just spent by the adversary. 
In Theorem~\ref{thm:ub-ee-s2} this adaptivity
happens transparently: the length of the $j$th message depends directly on the fraction $\beta_{j-1}$ of the error budget expended by the adversary in round $j-1$,
even though $\beta_{j-1}$ is not ever 
named as a parameter of the algorithm.

A key difference between $\ExistsEqual$ and $\EqualityTesting$ is that in the latter,
the adversary can effectively \emph{hide} how much of its error budget it has expended.
Consider the state of Bob after receiving the first message from Alice.   If 
he finds that $|I_1|=k/2$, there is no way to tell how many false positives are contained
in $I_1$ and how many are true positives. In the \emph{worst case} the number
of false positives could be as high as $k^{1-1/r}$.  We cannot optimistically
assume the false positive number is lower,\footnote{Invoking Lemma~\ref{lem:hamming-distance} 
with a Hamming distance $d$ that is too small can result in an undetected failure of the protocol.}
and continually using the pessimistic bound leads to $O(rEk^{1/r})$ communication.
It seems that any optimal algorithm must \emph{detect} and \emph{adapt to} the fraction 
of the error budget spent 
by the adversary.

\begin{theorem} \label{thm:ub-et-s3}
Fix any $k \ge 1$, $E = \Omega(1)$, and $r \in [1,(\log k_0)/2]$, where $k_0 \leq \min\{k,E\}$. 
There exists a randomized protocol for {\EqualityTesting} length-$k$ vectors $x,y$
with Hamming distance $\dist(x,y)\leq k_0$ that uses $r$ rounds, 
$O(k + Ek_0^{1/r}\log r + Er \log r)$ bits of communication,
and errs with probability $\perr = 2^{-(E+1)}$.
\end{theorem}

The remainder of this section constitutes a proof of Theorem~\ref{thm:ub-et-s3}.

Define $E'=7E$ to be the \emph{error budget} of the adversary, i.e., it is allowed to make up to $E'$ inner product tests pass on unequal coordinates.

\paragraph{Round 1.} Initially $I_0 = [k]$ is guaranteed to contain at most $k_0$ unequal coordinates.  
Alice generates $l_1 = E'k_0^{1/r-1}$ test bits for
each coordinate in $I_0$, and transmits them to Bob using Lemma~\ref{lem:hamming-distance} 
with a Hamming distance of $d=k_0$.
(We show later that the $k_0\log(k/k_0)$ terms in this protocol
contribute negligibly to the overall communication; 
thus, for the time being we measure the 
cost as $k_0l_1 = E'k_0^{1/r}$.)
Bob sets $I_1$ to be the subset of $I_0$ that pass all 
inner product tests.  

\paragraph{Round 2.} Due to the adversary's error budget, the
number of false positives in $I_1$ is at most 
\[
k_1 = E'/l_1 = k_0^{1-1/r}.
\]
Suppose the \emph{true} number of false positives in $I_1$ is
\[
k_1^* = \beta_1'k_1 = \beta_1'k_0^{1-1/r},
\]
meaning the adversary just spent a $\beta_1'$ fraction of his total error budget. 
Bob cannot measure $\beta_1'$, but he can send a message to Alice that allows \emph{her} to estimate $\beta_1'$.  
Bob invokes Lemma~\ref{lem:hamming-distance} $\log r$ times.  
For $i\in [1,\log r]$, Bob generates the next $l_2^{(i)}$ test bits
for coordinates in $I_1$ so that Alice can recover them up to
a Hamming distance of $k_1^{(i)}$, where
\begin{align*}
    l_2k_1 &= 2E'k_0^{1/r},\\
    l_2^{(i)} &= l_2\cdot 2^{i-1},\\
    k_1^{(i)} &= k_1 / 2^{i-1}.
\end{align*}
Clearly invocation $i$ will succeed if $k_1^{(i)} \ge k_1^*$ 
and may fail if $k_1^{(i)} < k_1^*$.  In order to detect
which invocations of Lemma~\ref{lem:hamming-distance} succeed, Bob
supplements each with a $\Theta(E)$-bit hash of the test bits generated.  Thus, with probability $1-2^{-\Theta(E)}$, Lemma~\ref{lem:hamming-distance} has no \emph{silent} failures.

\paragraph{Round 3 Onward.} Suppose that Alice detects 
that the invocations of Lemma~\ref{lem:hamming-distance}
with Hamming distances $k_1^{(1)},\ldots,k_1^{(i^*)}$ succeed
but the one with $k_1^{(i^*+1)}$ fails (or that $i^* = \log r$).  
Alice estimates $\beta_1'$ by 
\[
\beta_1 = 2^{-i^*}.
\]
Observe that if $i^* < \log r$ (the $(i^*+1)$th invocation of 
Lemma~\ref{lem:hamming-distance} fails)
then 
\[
k_1^* \geq k_1^{(i^*+1)} = 2^{-i^*}\cdot k_0^{1-1/r}
\]
and consequently, $\beta_1 \leq \beta_1'$.
On the other hand, if $i^* = \log r$ then $\beta_1 = 1/r$ whereas
$\beta_1'$ may be close to zero.  Either way, we have
\[
\beta_1 \leq \beta_1' + 1/r.
\]
Since every false positive in $I_2$ has successfully passed $l_2^{(i^*)}$
inner product tests, we can conclude that the maximum number of false positives
remaining in $I_2$ is
\[
k_2 = \frac{E'}{l_2^{(i^*)}} 
    = \frac{E'}{2E'k_0^{1/r}\cdot 2^{i^*-1} / k_1}
    = \frac{k_1}{k_0^{1/r} 2^{i^*}} = \beta_1 k_0^{1-2/r}.
\]
As before, $k_2^* = \beta_2'k_2$ is the true number of false positives in $I_2$, 
where $\beta_2'\in[0,1]$ is currently unknown.
In the third round Alice selects $l_3$ (see below), and invokes 
the test bit exchange protocol (Lemma~\ref{lem:hamming-distance})
with $l_3^{(i)} = l_32^{i-1}$ test bits per coordinate and Hamming distance
$k_2^{(i)} = k_2/2^{i-1}$, in parallel for all $i\in [1,\log r]$.  The overall communication volume
is $k_2l_3\log r$, and we select $l_3$ such that this is linear in the (estimated)
error budget spent by the adversary in round one, i.e.,
\[
k_2l_3\log r = \beta_1E'k_0^{1/r}\log r
\]
and therefore 
\[
l_3 = \beta_12E'k_0^{1/r}/k_2 = 2E'/k_0^{1-3/r}
\]
is independent of $\beta_1$. 

All the bounds above were specialized to round three, but apply to round $j$
by reindexing appropriately.  In particular, the receiver of the $(j-1)$th message
estimates $\beta_{j-2}' = k_{j-2}^*/k_{j-2}$ by
\begin{align*}
\beta_{j-2} &= 2^{-i^*} \leq \beta_{j-2}' + 1/r\\
\intertext{and sets}
k_{j-1} &= \beta_{j-2} k_0^{1-(j-1)/r},\\
l_j &= 2E'/k_0^{1-j/r},
\end{align*}
then invokes Lemma~\ref{lem:hamming-distance} $\log r$ times in parallel with parameters
$l_j^{(i)} = l_j2^{i-1}$ and $k_{j-1}^{(i)} = k_{j-1}/2^{i-1}$ to send the $j$th message.
It remains to bound the total communication and the probability of error.

\paragraph{Communication Volume.}
With the extra $\Theta(E)$-bit hash, the cost of each invocation of 
Lemma~\ref{cor:hash-test} with parameters $d,L$ is
$O(E + dL + d\log(k/d))$.  There are at most $\log r$ invocations per round
and $r$ rounds, so the total contributed by the first term is $O(Er\log r)$.  
The total contributed by the second term is:

\begin{align*}
    &\phantom{=} k_0l_1 + \sum_{j = 2}^{r} \sum_{i = 1}^{\log r} k_{j-1}^{(i)} l_j^{(i)} \\
    &= E'k_0^{1/r} + \log r\cdot \sum_{j = 2}^{r} k_{j-1} l_j \\
    &= E'k_0^{1/r} + 2E'k_0^{1/r}\log r\cdot\left(1 + \sum_{j = 3}^{r} \beta_{j-2}\right) \\
    &\leq E'k_0^{1/r} + 2E'k_0^{1/r}\log r\cdot\left(1 + \sum_{j = 3}^{r} (\beta_{j-2}'+1/r)\right) \\
    &\leq E'k_0^{1/r} + 6E'k_0^{1/r}\log r  & (\sum_{j\ge 3}\beta_{j-2}' < 1)\\ 
    &= E'k_0^{1/r}(6\log r+1).
\end{align*}

Next we bound the third term. For any $j \geq 2$, we have
\begin{align*}
    &\phantom{=}\sum_{i = 1}^{\log r} k^{(i)}_{j-1} \log (k/k^{(i)}_{j-1})\\
    &= \sum_{i=1}^{\log r}\frac{k_{j-1}}{2^{i-1}} \log \frac{2^{i-1}k}{k_{j-1}} \\
    &= k_{j-1} \log \frac{k}{k_{j-1}} \sum_{i=1}^{\log r} \frac{1}{2^{i-1}} + k_{j-1} \sum_{i=1}^{\log r} \frac{i-1}{2^{i-1}} \\
    &= O(k_{j-1} \log(k/k_{j-1})).
\end{align*}

Therefore, it suffices to only consider the first invocation of Lemma~\ref{lem:hamming-distance} from each round. 
Now we bound the total across all rounds. For the first two rounds we have $k_0 \log k/k_0 \leq k$ and $k_1 \log k/k_1 \leq k$, 
so we start counting from round $j = 3$.

\begin{align*}
    &\phantom{=} \sum_{j=3}^{r} k_{j-1} \log \frac{k}{k_{j-1}} \\
    &= \sum_{j=2}^{r-1} \beta_{j-1} k_0^{1-j/r} \log \frac{k}{\beta_{j-1} k_0^{1-j/r}} \\
    &\leq \sum_{j=2}^{r-1} \beta_{j-1} k \frac{\log \frac{k^{j/r}}{\beta_{j-1}}}{k^{j/r}} \\
    &\leq k\sum_{j=2}^{r-1} \frac{\log k^{j/r}}{k^{j/r}} + k \sum_{j=2}^{r-1} \frac{\beta_{j-1} \log \frac{1}{\beta_{j-1}}}{k^{j/r}} \tag{$\beta_j \leq 1$.}\\
    &\leq k\sum_{j=2}^{r-1} \frac{\log k^{j/r}}{k^{j/r}} + \frac{k}{e} \sum_{j=2}^{r-1} \frac{1}{k^{j/r}}  \\
    &= O(k). \tag{$k^{1/r} \geq 2$.}
\end{align*}
In conclusion, the total communication cost is
$O(k + Ek_0^{1/r}\log r + Er\log r)$.

\paragraph{Error Probability.}
We now show that protocol errs with probability less than $2^{-(E+1)}$.
If we use a $2E$-bit hash of the test bits in each invocation of Lemma~\ref{lem:hamming-distance}
the probability that any failed invocation goes unnoticed is at most
\[
r\log r\cdot 2^{-2E}.
\]
The algorithm works correctly so long as $k_j^* \leq k_j$ for every $j$,
which holds whenever the adversary does not exceed his \emph{error budget} $E'$. 
The probability that the error budget is exceeded is, by a union bound,
at most 
\[
\binom{E' + k_0 - 1}{k_0 - 1} 2^{-E'} \leq \exp(k_0 \log \left(e(E' + k_0)/k_0\right)-E'),
\]
which is less than $2^{-3E}$ when $E'=7E\geq 7k_0$.  Finally, every unequal coordinate
is ultimately subject to $l_r = 2E'$ inner product tests, and the probability that any 
goes undetected is at most $k_02^{-2E'}$.  The total error probability is therefore at most
\[
r\log r\cdot 2^{-2E} + 2^{-3E} + k_02^{-2E'} \ll 2^{-(E+1)}.
\]
This concludes the proof of Theorem~\ref{thm:ub-et-s3}.

\section{Distributed Triangle Enumeration}\label{sect:triangle}

One way to solve local triangle enumeration in the $\CONGEST$ model is to execute,
in parallel, a $\SetIntersection$ protocol across
every edge of the graph, where the set associated with a
vertex is a list of its neighbors.  Since there are at most
$\Delta n/2$ edges, we need the $\SetIntersection$ error probability to be $2^{-E}$, $E = \Theta(\log n)$, 
in order to guarantee a global success probability of $1-1/\poly(n)$.  Our lower bound says any algorithm taking
this approach must take 
$\Omega((\Delta + E\Delta^{1/r})/\log n + r)$ rounds since
each round of $\CONGEST$ allows for one $O(\log n)$-bit message.
The hardest situation seems to be when $\Delta = E = \Theta(\log n)$, in which case the optimum choice is to set $r=\log \Delta$, making the triangle enumeration algorithm run in $O(\log \Delta)=O(\log\log n)$ time.  In Theorem~\ref{thm:triangleDelta} we show that it is possible to handle this situation exponentially faster, in $O(\log\log \Delta)=O(\log\log\log n)$ time, 
and in general, to solve local triangle enumeration~\cite{IzumiG17} 
in optimal $O(\Delta/\log n)$ time
so long as $\Delta > \log n\log\log\log n$.

\begin{theorem}\label{thm:triangleDelta}
Local triangle enumeration can be solved in a $\CONGEST$
network $G=(V,E)$ with maximum degree $\Delta$
in $O(\Delta/\log n + \log\log\Delta)$ rounds
with probability $1-1/\poly(n)$.
This is optimal for all 
$\Delta = \Omega(\log n\log\log\log n)$.
\end{theorem}

\begin{proof}
The algorithm consists of 
$\min\{\log\log\Delta, \log\log\log n\}$ 
phases.  The goal of the first phase is to transform the original 
triangle enumeration problem into one with maximum degree $\Delta_1 < (\log n)^{o(1)}$, in $O(\log^* n)$ rounds of communication.  The goal of every subsequent phase is to reduce the maximum degree from $\Delta' \leq \sqrt{\log n}$ to $\sqrt{\Delta'}$, in $O(1)$ rounds of communication.  Thus, the total number of rounds is $O(\log\log\Delta)$ rounds if the first round is skipped, and
$O(\log^* n + \log\log(\Delta_1))=O(\log\log\log n)$ otherwise.

\paragraph{Phase One.} Suppose $\Delta \ge \sqrt{\log n}$.
Each vertex $u$ is identified with the set 
$A_u = \{\ID(v) \mid \{v,u\}\in E\}$ having size $\Delta$.  
For each $\{u,v\}\in E$ we reduce $\SetIntersection$
to $\EqualityTesting$ by applying Theorem~\ref{thm:Equivalence-Int-Eq}, then
run the two-party $\EqualityTesting$ 
protocol of Theorem~\ref{thm:ub-et-1},
with $k = \max\{\Delta,\log n\}, r = \log^* n$, and $E=r^{-1} k^{1-1/r}$.  
(I.e., if $\Delta<\log n$ we imagine padding each set to size $\log n$ with dummy elements.)
One undesirable property of this protocol is that it can fail
``silently'' if the preconditions of Lemma~\ref{lem:hamming-distance}
are not met.  When the Hamming distance between two strings
exceeds the threshold $d$, Bob generates a garbage string 
$x'\neq x$ but fails to detect this.  To rectify this problem, 
we change the Lemma~\ref{lem:hamming-distance} protocol
slightly: Alice sends the color $\phi(x)$ of her string, as well
as an $O(\log n)$-bit hash $h(x)$.  Bob reconstructs $x'$ as usual
and terminates the protocol if $h(x) \neq h(x')$.  Clearly the probability
of an undetected failure (i.e., $x\neq x'$ but $h(x)=h(x')$) 
is $1/\poly(n)$.  Define $G_1 = (V,E_1)$ such that $\{u,v\}\in E_1$
iff the $\SetIntersection$ protocol over $\{u,v\}$ detected a
failure.  In other words, with high probability, 
all triangles in $G$ have been discovered, 
except for those contained \emph{entirely} inside $G_1$.
The probability that any particular edge appears in $E_1$
is $2^{-E} = 2^{-k^{1-1/\log^* n}/\log^* n}$ and independent
of all other edges.  In particular, if $\Delta \gg (\log n)^{1+1/\log^* n}$
then no errors occur, with probability $1-1/\poly(n)$.  
Define $\Delta_1$ to be the maximum degree in $G_1$.  Thus,
\begin{align*}
    \Pr\left[\Delta_1 \geq (\log n)^{2\epsilon}\right] 
    &\leq n\cdot {\Delta \choose (\log n)^{2\epsilon}} \cdot \left(2^{-E}\right)^{(\log n)^{2\epsilon}} 
                                                            & \epsilon = 1/r = 1/\log^* n\\
    &\leq n\cdot \exp\left(O((\log n)^{2\epsilon}_{\ } \log\log n) \, - \epsilon (\log n)^{1-\epsilon}\cdot (\log n)^{2\epsilon}\right)\\
    &\leq 1/\poly(n).
\end{align*}

\paragraph{Phases Two and Above.} Suppose that at some round,
we have detected all triangles except for those contained
in some subgraph $G'=(V,E')$ having maximum degree 
$\Delta' < \sqrt{\log n}$.
Express $\Delta'$ as $(\log n)^{\gamma}$, where $\gamma<1/2$. 
We execute the $\EqualityTesting$  
protocol of Theorem~\ref{thm:ub-et-s2-1}
with $k=\Delta'$, $r=2$, and $E = C(\log n)^{1-\gamma/2}$ for a sufficiently large constant $C$.  Note that $1-\gamma/2 > \gamma$, 
so $E>k$, as required by Theorem~\ref{thm:ub-et-s2-1}.
The protocol takes
$O(Ek^{1/2}/\log n + r) = O(1)$ rounds since the 
communication volume is $O(Ek^{1/2}) = O(\log n)$ and $r=2$.
Let $G''$ be the subgraph of $G'$ consisting of 
edges whose protocols detected a failure and $\Delta''$ be the maximum degree in $G''$.  Once again,
\begin{align*}
    \Pr\left[\Delta'' \geq (\log n)^{\gamma/2}\right] 
    &\leq n\cdot {\Delta' \choose (\log n)^{\gamma/2}} \cdot \left(2^{-E}\right)^{(\log n)^{\gamma/2}} \\
    &\leq n\cdot \exp\left(O((\log n)^{\gamma/2}\log\log n)  \,- C(\log n)^{1-\gamma/2}\cdot (\log n)^{\gamma/2}\right)\\
    &\leq 1/\poly(n).
\end{align*}
Thus, once $\Delta \leq \sqrt{\log n}$,
$\log\log \Delta \leq \log\log\log n - 1$ 
of these 2-round phases suffice to 
find all remaining triangles in $G$.
\end{proof}

Theorem~\ref{thm:triangleDelta} depends critically on the duality between edges and $\SetIntersection$ instances,
and between edge endpoints and elements of sets.  In particular, when an execution of a $\SetIntersection$ over $\{u,v\}$ is successful, this effectively \emph{removes} 
$\{u,v\}$ from the graph, thereby removing many occurrences of $\ID(u)$ and $\ID(v)$ from adjacent sets.

Consider a slightly more general situation where we have a graph of \emph{arboricity} $\lambda$ (but unbounded $\Delta$), witnessed by a given acyclic orientation having out-degree at most $\lambda$.  Redefine the set $A_u$ to be the set of out-neighbors of $u$.
\[
A_u = \{\ID(v) \mid \{u,v\}\in E \mbox{ with orientation } u\rightarrow v\}.
\]
By definition $|A_u|{} \leq \lambda$.
Because the orientation is acyclic, every triangle on $\{x,y,z\}$ is (up to renaming) oriented as $x\rightarrow y,\; x\rightarrow z,\; y\rightarrow z$.  Thus, it will \emph{only} be detectable by the $\SetIntersection$ instance associated with $\{x,y\}$.

\begin{theorem}\label{thm:triangle-orientation}
Let $G=(V,E)$ be a $\CONGEST$ network equipped with an acyclic orientation with outdegree at most $\lambda$.
We can solve local triangle enumeration on $G$ in 
$O(\lambda/\log n + \log\lambda)$ time.
\end{theorem}

\begin{proof}
We apply Theorem~\ref{thm:Equivalence-Int-Eq}
to reduce each $\SetIntersection$ instance to
an $\EqualityTesting$ instance, then apply 
Theorem~\ref{thm:ub-et-2} with 
$E = \Theta(\log n)$ and $r=\log \lambda$
to solve each with $O(\lambda + E\lambda^{1/r}) = O(\lambda + E)$ communication in $O((\lambda + E)/\log n + r) = O(\lambda/\log n + \log\lambda)$ time.
Note that the dependence on $\lambda$ here
is exponentially worse
than the dependence on $\Delta$ in Theorem~\ref{thm:triangleDelta}.
\end{proof}

It may be that $G$ is known to have arboricity $\lambda$, but an acyclic orientation is unavailable.  The well known ``peeling algorithm'' (see~\cite{ChibaN85} or \cite{BarenboimE10}) computes a $C\lambda$ orientation in $O(\log_C n)$ time for $C$ sufficiently large, say $C\ge 3$.
Using this algorithm as a preprocessing step, we can solve local triangle enumeration optimally when $\lambda = \Omega(\log^2 n)$.

\begin{theorem}\label{thm:triangle-arboricity}
Let $G=(V,E)$ be a $\CONGEST$ network having arboricity $\lambda$ (with no upper bound on $\Delta$).  
Local triangle enumeration can be solved in 
optimal $O(\lambda/\log n)$ time 
when $\lambda = \Omega(\log^2 n)$,
and sublogarithmic time 
$O(\log n / \log(\log^2 n/\lambda))$ otherwise.
\end{theorem}

\begin{proof}
The algorithm computes a $\gamma\cdot \lambda$ orientation in $O(\log_\gamma n)$ time and then applies Theorem~\ref{thm:triangle-orientation} to solve
local triangle enumeration in 
$O(\gamma\lambda/\log n + \log(\gamma\lambda))$ time.
The only question is how to set $\gamma$.  
If $\lambda = \Omega(\log^2 n)$ we set $\gamma=3$,
making the total time $O(\lambda/\log n)$, which is optimal~\cite{IzumiG17}.   Otherwise we choose $\gamma$ to balance the $\log_\gamma n$ and $\gamma\lambda/\log n$ terms, so that
\[
\gamma\log\gamma = \log^2 n/\lambda
\]
Thus, the total running time is slightly sublogarithmic $O(\log n / \log(\log^2 n/\lambda))$.  Specifically, it is $O(\log n/\log\log n)$ whenever $\lambda < \log^{2-\epsilon} n$.
\end{proof}

\section{Conclusions and Open Problems}\label{sect:conclusion}

We have established a new three-way tradeoff between rounds, communication, and error probability for many fundamental problems in communication complexity such as $\SetDisjointness$ and $\EqualityTesting$.  Our lower bound is largely incomparable to the round-communication lower bounds 
of~\cite{SaglamT13,BrodyCKWY16}, and stylistically very different from both~\cite{SaglamT13} and \cite{BrodyCKWY16}.  We believe that our method \emph{can} be 
extended to recover \Saglam{} and Tardos's~\cite{SaglamT13} tradeoff (in the constant error probability regime), but with a more ``direct'' proof
that avoids some technical difficulties arising from their round-elimination technique.
It is still open whether $\EqualityTesting$ can be solved in $r$ rounds with precisely $O(Ek^{1/r})$ 
communication and error probability $2^{-E}<2^{-k}$.  Our algorithms match this lower bound
only when $r=O(1)$ or $r=\Omega(\log k)$, or for any $r$ when solving the easier $\ExistsEqual$ problem.

We developed some $\CONGEST$ algorithms for triangle enumeration that employ two-party $\SetIntersection$
protocols.  It is known that this strategy is suboptimal when $\Delta \gg n^{1/3}$~\cite{ChangPZ19,ChangS19}.  However, for the \emph{local} triangle enumeration problem,\footnote{Every triangle must be reported by one of its three constituent vertices.}
our $O(\Delta/\log n + \log\log\Delta)$ algorithm is optimal~\cite{IzumiG17} 
for every $\Delta = \Omega(\log n\log\log\log n)$.  
Whether there are faster algorithms for triangle \emph{detection}\footnote{At least one vertex must announce there is a triangle; there is no obligation to list them all.} 
is an intriguing open problem.  It is known that 1-round $\LOCAL$ algorithms
must send messages of $\Omega(\Delta\log n)$ bits 
deterministically~\cite{AbboudCKL17} or $\Omega(\Delta)$ bits randomized~\cite{FischerGKO18}.
Even for \underline{$2$-round} triangle detection algorithms,
there are no nontrivial communication lower bounds known.

\bibliographystyle{alpha}
\bibliography{ref}

\appendix

\section{Reductions and Near Equivalences}\label{sect:reductions}

Brody et al.~\cite{BrodyCKWY16} proved that $\SetIntersection$ on sets of size $k$ is reducible to $\EqualityTesting$ on vectors of length $O(k)$, at the cost of one round and $O(k)$ bits of communication.  
However, the reduction is \emph{randomized} and fails with 
probability at least $\exp(-\tilde{O}(\sqrt{k}))$.  This is the probability that when $k$ balls are thrown uniformly at random into
$k$ bins, some bin contains $\omega(\sqrt{k})$ balls.

Recall the statement of Theorem~\ref{thm:Equivalence-Int-Eq}:
\begin{align*}
\EQ(k,r,\perr) &\leq \SetI(k,r,\perr), 
& \SetI(k,r+1,\perr) &\leq \EQ(k,r,\perr) + \zeta,\\
\ExistsEQ(k,r,\perr) &\leq \SetD(k,r,\perr), 
&\SetD(k,r+1,\perr) &\leq \ExistsEQ(k,r,\perr) +\zeta,
\end{align*}
where $\zeta = O(k + \log\log\perr^{-1})$.
In other words, under any error regime $\perr$, 
the communication complexity of $\SetIntersection$ and $\EqualityTesting$ 
are the same, up to one round and $O(k + \log\log\perr^{-1})$ bits of communication,
and that the same relationship holds
between $\SetDisjointness$ and $\ExistsEqual$.  
The proof is inspired by the probabilistic reduction of 
Brody et al.~\cite{BrodyCKWY16}, but uses succinct encodings of 
perfect hash functions rather than random hash functions.

\begin{proof}[Proof of Theorem~\ref{thm:Equivalence-Int-Eq}]
The leftmost inequalities have been observed before~\cite{SaglamT13,BrodyCKWY16}.
Given inputs $x,y$ to $\ExistsEqual$ or $\EqualityTesting$, 
Alice and Bob generate sets $A = \{(1,x_1),\ldots,(k,x_k)\}$ and 
$B = \{(1,y_1),\ldots,(k,y_k)\}$ 
before the first round of communication
and then proceed to solve $\SetIntersection$ or $\SetDisjointness$ 
on $(A,B)$.
Knowing $A\cap B$ or whether $A\cap B = \emptyset$ clearly 
allows them to determine the correct output of $\EqualityTesting$ or $\ExistsEqual$ on $(x,y)$.

The reverse direction is slightly more complicated.
Let $(A,B)$ be the instance of $\SetIntersection$ or $\SetDisjointness$
over a universe $U$ with size at most $|U| {} =O(k^2/\perr)$.
Alice examines her set $A$, and picks a 
\emph{perfect} hash function 
$h : U \mapsto [k]$ for $A$,
i.e., $h$ is injective on $A$. 
(This can be done in $O(k)$ time, in expectation,
using only \emph{private} randomness.
In principle Alice could do this 
step deterministically, given sufficient time.)
Most importantly, $h$ can be described using 
$O(k + \log\log|U|) = O(k + \log\log\perr^{-1})$ 
bits~\cite{SchmidtS90}, using a variant
of the Fredman-\Komlos-\Szemeredi~\cite{FredmanKS84} 
2-level perfect hashing 
scheme.\footnote{We sketch how the encoding of $h$ works, for completeness.  First, pick
a function $h' : U \mapsto [O(k^2)]$ that is 
collision-free on $A$.  Fredman et al.~\cite{FredmanKS84}
proved that a function of the form 
$h'(x) = (ax \mod p) \mod O(k^2)$ works with constant probability,
where $p = \Omega(k^2\log|U|)$ is prime and $a\in [0,p)$ is random.
Pick another function $h_{*} : [O(k^2)]\mapsto [k]$ 
that has at most twice the expected
number of collisions on $A$, namely $2\cdot {k\choose 2}/k < k$, 
and partition $A$ into $k$ buckets $A_j = A\cap h_*^{-1}(j)$. 
The sizes $|A_0|, |A_1|, \ldots, |A_{k-1}|$
can be encoded with $2k$ bits.  
We now pick $O(\log k)$ pairwise independent 
hash functions $h_1, h_2, \ldots, h_{O(\log k)} : [O(k^2)]\mapsto [O(k^2)]$.  For each bucket $A_j$,
we define $h_{(j)}$ to be the function with the minimum 
$i$ for which $h_{(j)}(x) = h_i(x) \mod {} |A_j|^2$ is 
injective on $A_j$.  In order to encode which function
$h_{(j)}$ is (given that $h_1,\ldots,h_{O(\log k)}$ 
are fixed and that $|A_j|$ is known),
we simply need to write $i$ in unary, i.e., 
using the bit-string $0^{i-1}1$.  
This takes less than 2 bits per $j$ in expectation since
each $h_i$ is collision-free on $A_j$ with probability at least 1/2.  Combining $h',h_*, |A_0|,\ldots,|A_{k-1}|$ 
and $h_{(0)},\ldots,h_{(k-1)}$ into a single injective function from 
$U \mapsto [O(k)]$ is straightforward,
and done exactly as in~\cite{FredmanKS84}.  
By marking which elements in this range are actually 
used ($O(k)$ more bits), we can generate the perfect 
$h : U\mapsto [k]$ whose range has size
precisely $k$.
Encoding $h'$ takes $O(\log k+\log\log|U|)$ bits and 
encoding $h_*$ takes $O(\log k)$ bits. 
The distribution $|A_0|,\ldots,|A_{k-1}|$ can be encoded with $2k$ bits.
The functions $h_1,\ldots,h_{O(\log k)}$ can be encoded in $O(\log^2 k)$ bits,
and the functions $h_{(0)},\ldots,h_{(k-1)}$ with less than $2k$ bits in expectation.
}
Alice sends the $O(k+\log\log\perr^{-1})$-bit description of 
$h$ to Bob.
Bob calculates $B_j = B\cap h^{-1}(j)$
and responds to Alice with the distribution $|B_0|,|B_1|,\ldots,|B_{k-1}|$, 
which takes at most $2k$ bits.
They can now generate an instance of 
Equality Testing
where the $k$ equality tests are the 
pairs $A_0 \times B_0, A_1\times B_1, \ldots, A_{k-1}\times B_{k-1}$.
By construction, $A_j = A\cap h^{-1}(j)$ is a 1-element set.
There is clearly 
a 1-1 correspondence between equal pairs and elements in $A\cap B$.  
We have Bob speak first in the $\EqualityTesting$/$\ExistsEqual$ protocol; 
thus, the overhead for this reduction is just 1 
round of communication and $O(k+\log\log\perr^{-1})$ bits.
\end{proof}

\end{document}